\documentclass[12pt]{article}

\usepackage[left=1in, right=1in, top = 1in, bottom = 1in]{geometry}
\usepackage{amssymb}
\usepackage{amsmath}
\usepackage{amsthm}
\usepackage{bm}
\usepackage{bbm}
\usepackage{algorithm}
\usepackage{algpseudocode}
\usepackage{graphicx}
\usepackage{lineno}
\usepackage{authblk}
\usepackage[round]{natbib}
\usepackage{appendix}
\newtheorem{proposition}{Proposition}

\begin{document}

\title{Optimal prediction of positive-valued spatial processes: asymmetric power-divergence loss}
\date{}
\author[1]{Alan R. Pearse\thanks{Corresponding author: arp320@uowmail.edu.au}}
\author[1]{Noel Cressie}
\author[1]{David Gunawan}
\affil[1]{NIASRA, School of Mathematics and Applied Statistics, University of Wollongong, Northfields Avenue, Wollongong, New South Wales, 2522, Australia}

\maketitle
          
\begin{abstract}

This article studies the use of asymmetric loss functions for the optimal prediction of positive-valued spatial processes. We focus on the family of power-divergence loss functions due to its many convenient properties, such as its continuity, convexity, relationship to well known divergence measures, and the ability to control the asymmetry and behaviour of the loss function via a power parameter. The properties of power-divergence loss functions, optimal power-divergence (OPD) spatial predictors, and related measures of uncertainty quantification are examined. In addition, we examine the notion of asymmetry in loss functions defined for positive-valued spatial processes and define an asymmetry measure that is applied to the power-divergence loss function and other common loss functions. The paper concludes with a spatial statistical analysis of zinc measurements in the soil of a floodplain of the Meuse River, Netherlands, using OPD spatial prediction.

\end{abstract} 
\section{Introduction}\label{sec:introduction}

This article is based on the first author's PhD research under the supervision of the second and third authors. It follows on from a talk given by the second author at Spatial Statistics 2023: Climate and the Environment, in July 2023. Here, we concentrate on optimal prediction of positive-valued spatial processes using decision theory and a family of loss functions called the \textit{power-divergence loss functions} \cite[]{Cressie2022}. 

Let $Y(\cdot) \equiv \{Y(\mathbf{s}): \mathbf{s} \in D\}$ be a positive-valued spatial process defined at $d$-dimensional locations $\mathbf{s}$ in spatial domain $D$ (e.g., $D \subset \mathbb{R}^d$, the $d$-dimensional Euclidean space). Positive-valued spatial processes are examined here for two reasons. Firstly, many of the phenomena we study in nature or in society relate to positive (or at least non-negative)-valued processes. There are some exceptions, such as temperature and atmospheric trace-gas fluxes, but many of the things that scientists are interested in measuring are magnitudes, counts, or proportions (e.g., tree diameter at breast height, counts of a taxonomic group, or percent contamination), which are by definition positive or non-negative. Secondly, we seek a coherent non-approximative approach to statistical modelling of positive-valued spatial data and processes. Transformation-based approaches, such as trans-Gaussian kriging \cite[e.g.,][p. 127]{Cressie1993} and Gaussian-process modelling following a Box-Cox transformation \cite[]{Box1964}, result in spatial predictors whose statistical properties rely on likelihood and delta-method approximations. For example, the Box-Cox transformed process $X(\cdot) \equiv (Y(\cdot)^\nu-1)/\nu$, for $\nu \neq 0$, can only ever be a left-truncated Gaussian process, where the truncation occurs at $-\nu^{-1}$ \cite[]{FreemanModarres2006}. However, with some exceptions \cite[e.g.,][]{DeOliveira1997, ZM2016}, the transformed process is typically modelled as being Gaussian, whose support ranges from $-\infty$ to $\infty$. 

This article attempts to develop a coherent framework for prediction of positive-valued spatial processes within the framework of a spatial hierarchical statistical model. It takes a decision-theoretic approach, and it presents a loss function defined for positive-valued spatial processes that allows inference to be performed on the natural scale on which the processes are observed.   

A key goal of spatial statistics is to make inference on $Y(\cdot)$ based on a set of noisy, incomplete spatial data, $\mathbf{Z} \equiv (Z(\mathbf{s}_1), ..., Z(\mathbf{s}_n))'$, where $\mathbf{s}_1, ..., \mathbf{s}_n$ are a set of known observation locations. The term `inference' often refers to the estimation of parameters (including uncertainty quantification) on which the process may depend. Importantly in spatial statistics, it also refers to predicting a functional of the process $g(Y(\cdot))$ (including uncertainty quantification) based on data $\mathbf{Z}$. The functional $g(Y(\cdot))$ can refer to the value of $Y(\cdot)$ at a single prediction location, where $g(Y(\cdot)) = Y(\mathbf{s}_0)$ for any $\mathbf{s}_0 \in D$; or values at $M$ prediction locations, where $g(Y(\cdot)) = (Y(\mathbf{s}_{01}), ..., Y(\mathbf{s}_{0M}))'$ for $\mathbf{s}_{01}, ..., \mathbf{s}_{0M} \in D$; or average values over areas (or volumes) $B$ within $D$, where $g(Y(\cdot)) = Y(B) = |B|^{-1}\int_{B} Y(\mathbf{u}) \mathrm{d}\mathbf{u}$ for $B \subset D$ and $|B|\equiv \int_B \mathrm{d}\mathbf{u} >0$. We refer to $B$ as a `block.' 

In contemporary spatial statistics, it is recognised that the observations are an incomplete, noisy manifestation of the underlying scientific process. Then the prediction problem is solved using a spatial hierarchical statistical model \cite[e.g.,][]{CressieWikle2011}, which distinguishes between the observed spatial data $\mathbf{Z}$ and the underlying process $Y(\cdot)$. First, the data model $[\mathbf{Z} \mid Y(\cdot)]$ describes the distribution of the spatial data given the values of the underlying process, which characterises the distribution of the measurement error in the data. It is common to assume conditional independence of the data, $[\mathbf{Z}\mid Y(\cdot)] = [\mathbf{Z} \mid \mathbf{Y}] = \prod_{i=1}^n [Z(\mathbf{s}_i) \mid Y(\mathbf{s}_i)]$, where $\mathbf{Y} \equiv (Y(\mathbf{s}_1), ..., Y(\mathbf{s}_n))'$. Second, the distribution of the spatial process is described by the process model, $[Y(\cdot)]$. Finally, any parameters of the data model and process model are collected into the vector $\bm{\theta}$. The parameters may be estimated and plugged into the model or, if fully Bayesian inference is the goal, then a prior $[\bm{\theta}]$ is defined for the parameters (sometimes called the parameter model). The former approach has been called Empirical Hierarchical Modelling (EHM) and the latter approach Bayesian Hierarchical Modelling (BHM). In this article, our attention is focussed on inference on $Y(\cdot)$, conditional on $\mathbf{Z}$ and $\bm{\theta}$ (known or plugged in), and hence our hierarchical modelling is of the EHM type.

Consider spatial prediction of $Y(\mathbf{s}_0)$ and then allow $\mathbf{s}_0$ to vary over $D$, resulting in spatial prediction of the whole process $Y(\cdot)$. Prediction of $Y(\mathbf{s}_0)$ from spatial data $\mathbf{Z}$ follows from Bayes' rule: First, $[Y(\mathbf{s}_0), \mathbf{Y}]$ is obtained by marginalising the process model $[Y(\cdot)]$. Then recall the data model, $[\mathbf{Z}\mid Y(\cdot)] = \prod_{i=1}^n[Z(\mathbf{s}_i)\mid Y(\mathbf{s}_i)]$, and for simplicity assume $\mathbf{s}_0 \not\in \{\mathbf{s}_1, ..., \mathbf{s}_n\}$. Bayes' rule gives the posterior distribution,
$${[Y(\mathbf{s}_0), \mathbf{Y} \mid \mathbf{Z}] \propto \prod_{i=1}^n[Z(\mathbf{s}_i) \mid Y(\mathbf{s}_i)][Y(\mathbf{s}_0), \mathbf{Y}]}.$$ Integrating out $\mathbf{Y}$ gives the predictive distribution:
\begin{equation}
[Y(\mathbf{s}_0) \mid \mathbf{Z}] = \int [Y(\mathbf{s}_0), \mathbf{Y} \mid \mathbf{Z}]~\mathrm{d}\mathbf{Y}.\label{eqn:predictive_distribution}
\end{equation}
The predictive distribution \eqref{eqn:predictive_distribution} has all the information necessary to do spatial prediction of $Y(\mathbf{s}_0)$ (including uncertainty quantification), using the spatial data $\mathbf{Z}$.

What is the ``important'' information in the predictive distribution \eqref{eqn:predictive_distribution} for decision-makers? How will they distil this information into useful management outcomes? We seek summaries of the predictive distribution, one being the \textit{optimal predictor}, which we derive here from decision-theoretic considerations. 

Let $\delta(\mathbf{Z}; \mathbf{s}_0)$ be a predictor of $Y(\mathbf{s}_0)$. Then, define the \textit{loss function} $L(\delta(\mathbf{Z}; \mathbf{s}_0), Y(\mathbf{s}_0))$ to be a function that gives the penalty or loss incurred when $Y(\mathbf{s}_0)$ is predicted by $\delta(\mathbf{Z}; \mathbf{s}_0)$. The loss functions used for spatial prediction are examples of general decision-theoretic loss functions, $L: \mathcal{A} \times \mathcal{V} \mapsto [0, \infty)$, where $\mathcal{A}$ is the `action space' (space of possible actions) and $\mathcal{V}$ is the `value space' (space of possible values of the quantity of interest) \cite[]{Berger1985}. Here, the `action' is the predictor $\delta(\mathbf{Z}; \mathbf{s}_0)$, so the action space $\mathcal{A}$ is contained in the value space $\mathcal{V} = \{Y(\mathbf{s}_0): [Y(\mathbf{s}_0)] > 0\}$. These define the `prediction-observation domain' of \cite{Gneiting2011}. A loss function, $L(\delta(\mathbf{Z}; \mathbf{s}_0), Y(\mathbf{s}_0))$, must satisfy $L(\cdot, \cdot) \geq 0$, $L(Y(\mathbf{s}_0), Y(\mathbf{s}_0)) = 0$, and $E(L(\delta(\mathbf{Z}; \mathbf{s}_0), Y(\mathbf{s}_0)))$, taken with respect to the joint distribution $[Y(\mathbf{s}_0), \mathbf{Z}]$, must exist \cite[p. 3]{Berger1985}. 

Some widely known loss functions are given in Table \ref{table:common_loss_functions}. The optimal predictor of $Y(\mathbf{s}_0)$ under a particular loss function is derived by minimising with respect to $\delta$ the \textit{expected loss under the joint distribution} (ELJ), $[Y(\mathbf{s}_0), \mathbf{Z}]$, which is
\begin{equation}
    ELJ(\delta; \mathbf{s}_0) \equiv E(L(\delta(\mathbf{Z}; \mathbf{s}_0), Y(\mathbf{s}_0))).\label{eqn:defn_ELJ}
\end{equation} 
Hence, the optimal predictor, $\delta^*(\mathbf{Z}; \mathbf{s}_0)$, is
\begin{equation}
\delta^*(\mathbf{Z}; \mathbf{s}_0) \equiv \mathop{\text{arg~inf}}_{\delta} ELJ(\delta; \mathbf{s}_0) = \mathop{\text{arg~inf}}_{\delta} E\!\left\{L(\delta(\mathbf{Z}; \mathbf{s}_0), Y(\mathbf{s}_0))\right\};\label{eqn:ELJ}
\end{equation}
this is sometimes called the `Bayes risk.' By the law of iterated expectations, $\delta^*$ can also be obtained by minimising the \textit{expected loss under the predictive distribution} (ELP), $[Y(\mathbf{s}_0)\mid\mathbf{Z}]$, which is 
\begin{equation}
    ELP(\delta; \mathbf{Z}; \mathbf{s}_0) \equiv E(L(\delta(\mathbf{Z}; \mathbf{s}_0), Y(\mathbf{s}_0))\mid \mathbf{Z});\label{eqn:defn_ELP}
\end{equation}
this is sometimes called the `risk.' Hence, the left-hand side of \eqref{eqn:ELJ} can also be derived as,
\begin{equation}
\delta^*(\mathbf{Z}; \mathbf{s}_0) \equiv \mathop{\text{arg~inf}}_{\delta}ELP(\delta; \mathbf{Z};\mathbf{s}_0) = \mathop{\text{arg~inf}}_{\delta} E\!\left\{L(\delta(\mathbf{Z}; \mathbf{s}_0), Y(\mathbf{s}_0))\mid \mathbf{Z}\right\}.\label{eqn:ELP}
\end{equation}
Using the loss functions in Table \ref{table:common_loss_functions}, we obtain: for squared-error loss (SEL), $\delta_{SEL}^*(\mathbf{Z};\mathbf{s}_0) = E(Y(\mathbf{s}_0)\mid\mathbf{Z})$; for absolute-error loss (AEL), $\delta_{AEL}^*(\mathbf{Z}; \mathbf{s}_0)$ is the median of $[Y(\mathbf{s}_0)\mid \mathbf{Z}]$; for absolute relative loss (ARL), $\delta_{ARL}^*(\mathbf{Z};\mathbf{s}_0)$ is the median of the weighted predictive distribution, $[Y(\mathbf{s}_0)\mid\mathbf{Z}]/Y(\mathbf{s}_0)$ \cite[see also the generalisation of this in][]{Gneiting2011}; and for quantile loss (QTL), $\delta_{QTL, q}^*(\mathbf{Z}; \mathbf{s}_0)$ is the $q$-th quantile of $[Y(\mathbf{s}_0)\mid\mathbf{Z}]$. 

\begin{table}[h!]
\centering
\begin{tabular}{|p{0.35\textwidth}p{0.6\textwidth}|} 
 \hline
 Name & Definition\\ 
 \hline
 Squared-error loss (SEL) & $L_{SEL}(\delta, Y) = (\delta - Y)^2$\\ 
 Absolute-error loss (AEL) & $L_{AEL}(\delta, Y) = |\delta - Y|$ \\
 Absolute relative loss (ARL) & $L_{ARL}(\delta, Y) = |(\delta - Y)/Y|$\\
 Quantile loss (QTL) & $L_{QTL, q}(\delta, Y) = (\delta - Y)\{I(\delta - Y > 0) - q\};~q \in (0, 1)$ \\  
 \hline
\end{tabular}
\caption{Common loss functions. The spatial notation and dependence of the predictor on the data $\mathbf{Z}$ have been dropped for convenience.}\label{table:common_loss_functions}
\end{table}

The connection between spatial prediction and decision theory has been known for some time \cite[pp. 107-109]{Cressie1993}. The possibility of generalising Bayesian inference using different loss functions has been raised before, specifically with reference to information entropy and the asymmetric linear-exponential (LINEX) loss function \cite[]{Varian1975, Zellner1986}. In the context of spatial prediction, the decision-theoretic approach was further developed by \cite{Cressie2022} and \cite{Cressie2023}. However, despite the availability of other loss functions, the special case of squared-error loss (SEL) remains by far the most common choice of loss function for spatial prediction, regardless of whether there are distributional asymmetries in $[Y(\cdot)]$. This may be because of its many attractive statistical properties: If we choose SEL, the ELJ in \eqref{eqn:defn_ELJ} can be interpreted as the mean-squared prediction error (MSPE) of the predictor, and the optimal predictor is the mean of the predictive distribution, $\delta_{SEL}^*(\mathbf{Z}; \mathbf{s}_0) = E(Y(\mathbf{s}_0) \mid \mathbf{Z})$, which is unbiased in the sense that $E(\delta_{SEL}^*(\mathbf{Z}; \mathbf{s}_0)) = E(Y(\mathbf{s}_0))$. 

Alternatively, when it is difficult to obtain the predictive distribution for a spatial model (e.g., for non-Gaussian spatial processes), practitioners of spatial statistics often use kriging systems \cite[]{Matheron1962} for spatial prediction, which are based on linear combinations of the data that minimise the MSPE. An ability to sample from (rather than derive) the predictive distribution changes the emphasis in spatial statistics away from analytical formulae for prediction, and it questions the automatic use of SEL to carry out spatial prediction.

Loss functions are meant to codify the consequences of over-predicting and under-predicting $Y(\mathbf{s}_0)$. Choosing SEL may impart some attractive statistical properties on the optimal predictor but, from a decision-theoretic viewpoint, SEL should only be considered when the losses incurred from over-prediction and from under-prediction are the same. When predicting positive natural processes, we argue that the choice of SEL for spatial prediction should not be the default. 

Consider a scenario where $Y(\mathbf{s}_0)$ represents the peak height of floodwater in a town next to a river during a period of heavy rain. A natural disaster coordinator may be waiting on predictions of peak height to make decisions regarding whom and when to evacuate, as well as the number of sandbags and other preparatory measures that are needed to minimise water damage to key buildings. In this situation, under-prediction of the true peak height is far more costly than over-prediction. The former may lead to loss of life and significant damage to property. By comparison, the latter may lead to unnecessary inconvenience to residents of evacuated areas and excess expenditure on preparatory measures. Providing the disaster coordinator with a peak-height prediction based on SEL may therefore lead to worse outcomes than if one provided a positively-biased optimal predictor (i.e., over-prediction) based on an asymmetric loss function. When using asymmetric loss functions, the interpretation and properties of the resulting optimal spatial predictors may be unfamiliar or \textit{prima facie} undesirable (e.g., biasedness), but there are clear-use cases for them. There is a rich body of statistical theory to be uncovered by studying asymmetric loss functions. 

The aim of this article is to explore a class of asymmetric loss functions, the power-divergence loss (PDL) functions $\{L_{PDL,\lambda}: -\infty < \lambda < \infty\}$ defined in the next section, as a potentially useful family of loss functions for predicting positive spatial processes. Introduced in \cite{Cressie2022}, PDLs, $L_{PDL,\lambda}: \mathcal{A} \times \mathcal{V} \mapsto [0,\infty)$ for power parameter $\lambda \in (-\infty,\infty)$, have several useful properties: they are continuous over $\mathcal{A}\times \mathcal{V}$; they are connected to many well-known divergence measures; and the degree of asymmetry in the loss function and bias in the optimal predictor can be tuned via the power parameter $\lambda$. In what follows, we generalise many aspects of spatial prediction found when using SEL. 

The rest of this article is organised as follows. Section \ref{sec:asymmetry} presents the power-divergence loss function and a novel method of characterising asymmetry in loss functions. Section \ref{sec:prediction} discusses optimal prediction and uncertainty quantification under PDL, and Section \ref{sec:log_normal} applies this to the example of a log-Gaussian spatial process. Section \ref{sec:prediction_intervals} discusses conditional and unconditional prediction intervals based on PDL as a special form of uncertainty quantification. Section \ref{sec:Meuse} gives a practical demonstration of these results by applying optimal spatial prediction with PDL (including uncertainty quantification), to a dataset of zinc measurements in the soil of a floodplain of the Meuse River in the Netherlands \cite[]{Rikken1993}. Finally, Section \ref{sec:discussion} presents a discussion of our  results and future directions for research on spatial prediction based on power-divergence and other asymmetric loss functions.  

\section{Asymmetric loss functions for positive-valued spatial processes}\label{sec:asymmetry}

Recall that $\{Y(\mathbf{s}): \mathbf{s}\in D\}$ is a process in the class of all positive-valued spatial processes. The class is closed under the operation of re-scaling; that is, if $\{\nu(\mathbf{s}): \mathbf{s} \in D\}$ is any spatially-varying positive-valued function on $D$, then $\{\nu(\mathbf{s})\cdot Y(\mathbf{s}): \mathbf{s}\in D\}$ is also a positive-valued spatial process. (This re-scaling of the original process corresponds to a shifting of the logarithm of the process.) In what follows, we want to predict the original positive process directly from observations whose measurement errors are not additive. 

\subsection{Power-divergence loss functions}

The focus of this article is the class of power-divergence loss (PDL) functions \cite[]{Cressie2022}, which depends on the power parameter $\lambda \in (-\infty,\infty)$ as follows: For $-\infty < \lambda < \infty$,
\begin{align}
&L_{PDL,\lambda}(\delta(\mathbf{Z};\mathbf{s}_0), Y(\mathbf{s}_0)) = \nonumber\\ 
&\begin{cases}
\frac{1}{\lambda(\lambda + 1)}\left\{Y(\mathbf{s}_0)\left(\left(\frac{Y(\mathbf{s}_0)}{\delta(\mathbf{Z}; \mathbf{s}_0)}\right)^{\lambda} - 1\right) + \lambda(\delta(\mathbf{Z}; \mathbf{s}_0) - Y(\mathbf{s}_0))\right\},& \lambda \neq 0,-1\\
Y(\mathbf{s}_0) \log\left(\frac{Y(\mathbf{s}_0)}{\delta(\mathbf{Z}; \mathbf{s}_0)}\right) - (Y(\mathbf{s}_0) - \delta(\mathbf{Z}; \mathbf{s}_0)),& \lambda = 0\\
Y(\mathbf{s}_0) - \delta(\mathbf{Z}; \mathbf{s}_0) - \delta(\mathbf{Z}; \mathbf{s}_0) \log\left(\frac{Y(\mathbf{s}_0)}{\delta(\mathbf{Z}; \mathbf{s_0})}\right),& \lambda = -1.
\end{cases} \label{eqn:power_divergence_loss}
\end{align}
When $\lambda \neq 0, -1$, an equivalent expression is,
$$
L_{PDL,\lambda}(\delta(\mathbf{Z}; \mathbf{s}_0), Y(\mathbf{s}_0)) = \left(\frac{Y(\mathbf{s}_0)^{\lambda + 1} - (\lambda + 1)Y(\mathbf{s}_0)\delta(\mathbf{Z}; \mathbf{s}_0)^\lambda + \lambda\delta(\mathbf{Z}; \mathbf{s}_0)^{\lambda + 1}}{\lambda(\lambda + 1)\delta(\mathbf{Z};\mathbf{s}_0)^\lambda}\right).
$$
The PDL in \eqref{eqn:power_divergence_loss} is a continuous and convex function of $(\delta(\mathbf{Z}; \mathbf{s}_0), Y(\mathbf{s}_0))$ for $\delta(\mathbf{Z}; \mathbf{s}_0) >0$ and $Y(\mathbf{s}_0) > 0$. The key feature of power-divergence loss is that it is generally asymmetric, with the degree and direction of the asymmetry being controlled by $\lambda$ (Fig. \ref{fig:PDL}). Situations where over-prediction (under-prediction) is more serious than under-prediction (over-prediction) can be captured with PDL, as shown in the examples of Fig. \ref{fig:PDL}. 

\begin{figure}[!ht]
\centering
\includegraphics[width=\textwidth]{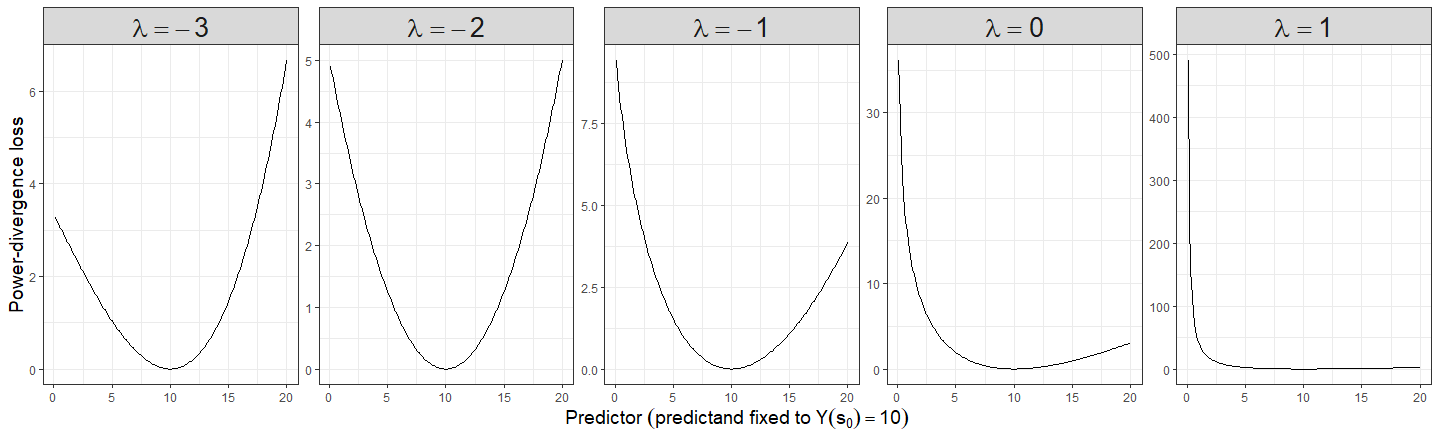}
\caption{Power-divergence loss (PDL) for $Y(\mathbf{s}_0) = 10$ and $\delta(\mathbf{Z}; \mathbf{s}_0) \in (0, 20]$, plotted for $\lambda \in \{-3, -2, -1, 0, 1\}$.}
\label{fig:PDL}
\end{figure}

PDL was motivated by \cite{Cressie2022} through goodness-of-fit statistics for discrete probability distributions, including the log-likelihood ratio and the Pearson chi-squared statistics as special cases \cite[]{Read1988}. Subsequently, power-divergence goodness-of-fit statistics were shown to be special cases in the class of phi-divergence statistics, and their properties were investigated by \cite{Cressie2002}. These statistics were based on divergences between two probability distributions \cite[including the Kullback-Leibler divergence as a special case;][]{Kullback1951}, but they can be generalised to divergences between two measures. For $x \in [0,1]$, define
$$
\phi_\lambda(x) \equiv \begin{cases}
\frac{1}{\lambda(\lambda + 1)}\left\{x^{\lambda + 1} - x\right\}, & \lambda \neq 0, -1\\
x\log(x), & \lambda = 0\\
-\log(x), & \lambda = -1.
\end{cases}
$$
Now extend the domain from $[0, 1]$ to $[0, \infty)$ as follows: For $x \in [0,\infty)$, define
$$
\phi_\lambda^+(x) \equiv \phi_\lambda(x) - \phi'_\lambda(1)(x - 1).
$$
Now, $\phi^+_\lambda(x)$ is continuous, convex on $[0, \infty)$, and satisfies $\phi_\lambda^+(1) = (\phi_\lambda^{+})'(1) = 0$, $(\phi_\lambda^{+})''(1) > 0$, $\phi_\lambda^+(0/0) = 0$, and $0\cdot \phi_\lambda^+(p/0) = p \cdot \lim_{u\rightarrow \infty}\phi_\lambda^+(u)/u$. Importantly, the loss function given by \eqref{eqn:power_divergence_loss} can be written as,
\begin{equation}
L_{PDL,\lambda}(\delta(\mathbf{Z}; \mathbf{s}_0), Y(\mathbf{s}_0)) = \delta(\mathbf{Z}; \mathbf{s}_0)\phi^+_\lambda\!\left(\frac{Y(\mathbf{s}_0)}{\delta(\mathbf{Z}; \mathbf{s}_0)}\right),   
\end{equation}
which (up to a scaling factor) emphasises its dependence on the ratio of predictand to predictor.

\subsection{Quantifying the asymmetry in loss functions}\label{sec:asymmetry_measure}

Recall that the process to be predicted is positive and, hence, it could be thought of as $Y(\cdot) = \exp\{X(\cdot)\}$, where the random process $X(\mathbf{s}) \in (-\infty, \infty)$ for any $\mathbf{s}\in D$. An obvious example is the log-Gaussian process $Y(\cdot)$, where $X(\cdot)$ is a Gaussian process, but our setting is more general; for example, $X(\cdot)$ could be an infinitely-divisible process. Misspecification could be represented by $X(\cdot)\pm c$, for $c > 0$, which corresponds to $\nu Y(\cdot)$ and $(1/\nu) Y(\cdot)$, where $\nu = \exp\{- c\} < 1$. Now write $\nu = 1 - f$, for $0 \leq f \leq 1$, and $f \times 100\%$ corresponds to the percentage decrease of $Y(\cdot)$ that results in the under-predicted value, $(1-f)Y(\cdot)$. Hence, $1/\nu = (1- f)^{-1} > 1$, and $(1-f)^{-1}Y(\cdot)$ represents the over-predicted value. In what follows we quantify the degree of asymmetry of loss functions by comparing losses due to multiplicative perturbations that define under-prediction and over-prediction, namely $Y(\mathbf{s}_0)$ perturbed to $(1-f)Y(\mathbf{s}_0)$ and $(1-f)^{-1}Y(\mathbf{s}_0)$ for $f \in (0, 1)$. Then, define the following measure of asymmetry,
\begin{equation}
A(f) \equiv \frac{L((1-f)Y(\mathbf{s}_0), Y(\mathbf{s}_0))}{L((1-f)^{-1}Y(\mathbf{s}_0), Y(\mathbf{s}_0))},~\text{for}~f\in(0, 1).\label{eqn:asymmetry_measure}
\end{equation}
The notation used in the left-hand side of \eqref{eqn:asymmetry_measure} might include dependence on $Y(\mathbf{s}_0)$ but in fact, for PDL and the loss functions in Table \ref{table:common_loss_functions}, it does not. A value of $1$ for $A(f)$ indicates that the loss function is `symmetric', and values of $A(f) > 1$ indicate that under-prediction of $Y(\mathbf{s}_0)$ incurs a higher loss than over-prediction, and conversely if $A(f) < 1$. The asymmetry measure \eqref{eqn:asymmetry_measure} for common loss functions, including SEL, AEL, ARL, and QTL (definitions given in Table \ref{table:common_loss_functions}) are presented in Table \ref{table:asymmetry_measures}. Note here that SEL and AEL, both classically considered to be `symmetric' loss functions for additive perturbations, are asymmetric for multiplicative perturbations, which is the natural type of perturbation for positive processes.

\begin{table}[h!]
\centering
\begin{tabular}{|c c c |} 
 \hline
 Loss function & Asymmetry & Value\\ 
 \hline
 SEL & $A_{SEL}(f)$ & $(1-f)^2$\\ 
 AEL & $A_{AEL}(f)$ & $1-f$\\
 ARL & $A_{ARL}(f)$ & $1-f$\\
 QTL & $A_{QTL, q}(f)$ & $(1-f)q/(1-q)$ \\  
 \hline
\end{tabular}
\caption{Asymmetry measures of common loss functions (see \eqref{eqn:asymmetry_measure}), where $f\in(0, 1)$. SEL stands for squared-error loss, AEL for absolute-error loss, ARL for absolute relative-error loss, and QTL for quantile loss.}
\label{table:asymmetry_measures}
\end{table}

The asymmetry measure for PDL simplifies as follows,
\begin{align}
A_{PDL,\lambda}(f) &= \frac{(1-f)^2\phi^+_{\lambda}\left((1-f)^{-1}\right)}{\phi^+_{\lambda}\left(1-f\right)} \nonumber\\
&= 
\begin{cases} 
\frac{(1-f)^{1-\lambda} - (1 - f) (1 + \lambda f) }{(1-f)^{\lambda + 1} + (\lambda + 1)f - 1}, &  \lambda \neq -1, 0\\
& \\
-\frac{(1-f)\log(1 - f) + f(1-f)}{(1-f)\log(1 - f)+f}, & \lambda = 0\\
&\\
-\frac{(1-f)^2\log(1-f) + f(1-f)}{\log(1-f)+f}, & \lambda = -1.
\end{cases}\label{eqn:A_PDL}
\end{align}
From Fig. \ref{fig:PDL}, $L_{PDL,\lambda}$'s asymmetry visually differs as $\lambda$ varies, but notice that $Y(\mathbf{s}_0)$ is fixed in those plots. The asymmetry measure \eqref{eqn:asymmetry_measure} is based on scaling perturbations of $Y(\mathbf{s}_0)$, where $Y(\mathbf{s}_0)$ is allowed to vary. From \eqref{eqn:A_PDL}, the asymmetry measure $A_{PDL, \lambda}(f)$ has the following properties: for $\lambda = 1$, the loss function is symmetric for all $f \in (0, 1)$; for all $\lambda < 1$, over-prediction is more serious than $f\times100\%$ under-prediction; and conversely for all $\lambda > 1$. The limit of $A_{PDL,\lambda}(f)$ as $f \rightarrow 0$ is always $1$; and as $f \rightarrow 1$, $A_{PDL,\lambda}(f)$ tends to $\infty$ if $\lambda > 1$, is equal to $0$ if $\lambda < 1$, and is equal to $1$ if $\lambda = 1$. See Section \ref{sec:asymmetry_measure_result} of the Online Supplement for details of the analysis.

\begin{figure}[!ht]
\centering
\includegraphics[width=\textwidth]{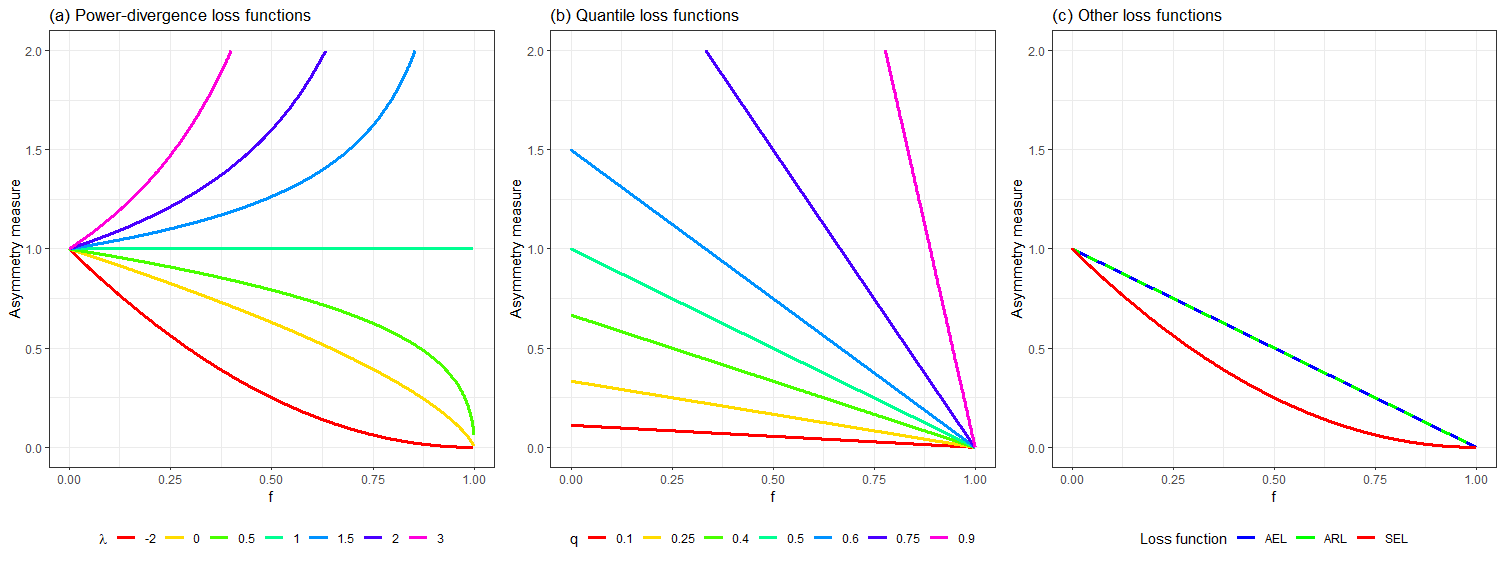}
\caption{Asymmetry measures plotted against $f \in (0, 1)$ for (a) power-divergence loss (PDL), (b) quantile loss (QTL), and (c) squared-error loss (SEL), absolute relative-error loss (ARL), and absolute error loss (AEL). For PDL, $\lambda \in \{-2, 0, 0.5, 1, 1.5, 2, 3\}$. For QTL, $q \in \{0.1, 0.25, 0.4, 0.5, 0.6, 0.75, 0.9\}$. The vertical axis is restricted to the range $[0, 2]$ to facilitate plotting and visual comparisons between loss functions. Any lines that touch the top edge of the plot continue past it.}
\label{fig:asymmetry}
\end{figure}

The asymmetry measure \eqref{eqn:asymmetry_measure} of different loss functions (i.e., PDL, QTL, SEL, AEL, and ARL) are compared in Fig. \ref{fig:asymmetry}. From \eqref{eqn:A_PDL} and Table \ref{table:asymmetry_measures}, $A_{PDL, -2}(f) = A_{SEL}(f) = (1-f)^2$ and $A_{PDL,-1/2}(f) = A_{QTL, 0.5}(f) = A_{AEL}(f) = A_{ARL}(f) = 1-f$. Fig. \ref{fig:asymmetry} shows that PDL can exhibit a variety of different behaviors; however, the common loss functions in Table \ref{table:asymmetry_measures} do not exhibit $A(f) > 1$ for any $f \in (0, 1)$. Further, Fig. \ref{fig:asymmetry} shows that the behaviour of QTL is fundamentally different from that of PDL except when $q = 0.5$. When the asymmetry measure is plotted for $L_{QTL, q}$, for various $q \in (0, 1)$, it shows a series of rays meeting at $(f, A_{QTL,q}(f)) = (1, 0)$, and $\lim_{f\rightarrow 0} A_{QTL, q}(f) = q/(1-q)$. Conversely, when PDL is plotted for several values of $\lambda \in (-\infty, \infty)$, the plot shows a series of rays radiating out from $(f, A_{PDL,\lambda}(f)) = (0, 1)$ and terminating at either $(1, 0)$, $(1, 1)$, or $(1, \infty)$. This illustrates that the family of quantile losses is fundamentally different from the family of power-divergence losses. Further, unlike PDL, which is a smooth, convex function that is continuous at $\delta(\mathbf{Z}; \mathbf{s}_0) = Y(\mathbf{s}_0)$, the piecewise linear nature of QTL means $A_{QTL, q}(f)$ has a discontinuity at the origin except when $q = 0.5$.

\section{The optimal power-divergence (OPD) spatial predictor}\label{sec:prediction}

Recall that an optimal predictor under a given loss function is derived by minimising either the ELP given by \eqref{eqn:defn_ELP} or equivalently the ELJ given by \eqref{eqn:defn_ELJ} over a family of predictors. Here, it is assumed that the family of predictors is the set of all measurable functions of the data, but it is not uncommon for restrictions to be imposed (e.g., the kriging predictor is obtained by minimising ELJ based on the SEL function (i.e., by minimising the MSPE) over the set of all linear functions of the data, subject to unbiasedness).

For PDL, the ELP is a function of predictor $\delta$ that we notate as,
\begin{equation}
ELP_\lambda(\delta; \mathbf{Z}; \mathbf{s}_0) \equiv E\{L_{PDL,\lambda}(\delta(\mathbf{Z}; \mathbf{s}_0), Y(\mathbf{s}_0)) \mid \mathbf{Z}\},\label{eqn:ELP_PDL}
\end{equation}
where the power parameter $\lambda$ appears as a subscript. Likewise the ELJ is notated as,
\begin{equation}
ELJ_\lambda(\delta; \mathbf{s}_0) \equiv E\{L_{PDL,\lambda}(\delta(\mathbf{Z};\mathbf{s}_0), Y(\mathbf{s}_0))\} \equiv E\{ELP_\lambda(\delta;\mathbf{Z}; \mathbf{s}_0)\}.\label{eqn:ELJ_PDL}
\end{equation}
The optimal power-divergence (OPD) spatial predictor of $Y(\mathbf{s}_0)$ is the predictor $\delta_\lambda^*(\mathbf{Z}; \mathbf{s}_0)$ that satisfies,
\begin{equation}
\delta_\lambda^*(\mathbf{Z}; \mathbf{s}_0) \equiv \mathop{\mathrm{arg~inf}}_{\delta} ELJ_\lambda(\delta; \mathbf{s}_0),\label{eqn:OPD_definition}
\end{equation}
where $L_{PDL,\lambda}(\delta(\mathbf{Z};\mathbf{s}_0), Y(\mathbf{s}_0))$ in \eqref{eqn:ELJ_PDL} is defined in \eqref{eqn:power_divergence_loss}, and the predictor $\delta$ belongs to the set of all measurable functions of the data. Solving \eqref{eqn:OPD_definition} yields \cite[]{Cressie2022}:
\begin{equation}
\delta_\lambda^*(\mathbf{Z}; \mathbf{s}_0) = \begin{cases} 
E\!\left(Y(\mathbf{s}_0)^{\lambda + 1} \mid \mathbf{Z}\right)^{1/(\lambda + 1)}& \lambda \neq 0, -1, \\
E(Y(\mathbf{s}_0)\mid \mathbf{Z})& \lambda = 0, \\
\exp\{E(\log(Y(\mathbf{s}_0)) \mid \mathbf{Z})\}& \lambda = -1.
\end{cases}\label{eqn:OPD_predictor}
\end{equation}
Interestingly, when $\lambda = 0$, the OPD spatial predictor is the mean of  the predictive distribution $[Y(\mathbf{s}_0) \mid \mathbf{Z}]$. However, the loss function $L_{PDL,0}(\delta, Y) = Y\log(Y/\delta) - (Y - \delta)$ is very different from the SEL function $(\delta - Y)^2$. In general, the result given by \eqref{eqn:OPD_predictor} shows that the OPD spatial predictor is a function of the $(\lambda+1)$-th fractional moment of $[Y(\mathbf{s}_0)\mid\mathbf{Z}]$. 

The existence of the $(\lambda + 1)$-th fractional moment of $[Y(\mathbf{s}_0)\mid\mathbf{Z}]$ is implied by the existence of the expectation of the PDL function. Recall that a property of any loss function used to make decisions is that the ELJ, $E(L_{PDL,\lambda}(\delta(\mathbf{Z}; \mathbf{s}_0), Y(\mathbf{s}_0)))$, must exist and be finite \cite[p. 3]{Berger1985}. Consequently, the ELP,  $E(L_{PDL,\lambda}(\delta;\mathbf{Z}; \mathbf{s}_0, Y(\mathbf{s}_0))\mid\mathbf{Z})$, must also exist and be finite. Hence, from \eqref{eqn:ELP_PDL}, it is easy to see that $E(Y(\mathbf{s}_0)^{\lambda + 1}\mid \mathbf{Z}) < \infty$.

In what follows, the properties of the OPD spatial predictor are explored. In this section, we only consider point-prediction. Section \ref{sec:prediction_intervals} deals with interval-prediction. 

\subsection{Computing the OPD spatial predictor}\label{sec:calculation}

Although \eqref{eqn:OPD_predictor} appears to be mathematically simple, computing it is not necessarily straightforward. The $(\lambda + 1)$-th moment involved will not be available in closed form for many predictive distributions (but see Section \ref{sec:log_normal}, where the log-Gaussian process is considered).  Some strategies for calculating the OPD spatial predictor in such cases are discussed here. 

\subsubsection*{Exact calculation and the moment generating function (MGF)}

For some predictive distributions, the OPD spatial predictor can be calculated exactly with very little effort (e.g., when $[Y(\mathbf{s}_0)\mid\mathbf{Z}]$ is a log-Gaussian distribution, a gamma distribution, or an inverse gamma distribution). We develop this further in Section \ref{sec:log_normal}. 

Even if there are no easy or obvious ways to get the OPD spatial predictor in closed form, it may still be possible to obtain if the moment generating function (MGF) of the predictive distribution exists and is known \cite[]{Cressie1986}. In well known classical theory, the MGF of a random variable $X$ exists if the expectation $M_X(t) = E\{\exp(tX)\}$ is finite for $t$ in some interval about the origin. If the MGF exists, then for positive random variables, \cite{Cressie1981} developed a method for obtaining \textit{all integer moments}, positive and negative. In that case, the well known result that differentiation of the MGF yields positive integer moments was generalised to obtain negative integer moments through integration. \cite{Cressie1986} gave the general result that all real moments can be obtained by fractional differentiation of the MGF.
In contemporary spatial statistics, it is more likely that a Monte Carlo method such as Markov Chain Monte Carlo (MCMC) will be available to sample from the predictive distribution. We now discuss such computational methods that can yield the OPD spatial predictor.

\subsubsection*{Computational methods}

In general, predictive distributions from EHMs and BHMs are not known in closed form. They are often obtained via Monte Carlo sampling such as MCMC. Let $\mathbf{y} \equiv (y_{(1)}, ..., y_{(M)})^\top$ be an $M$-dimensional vector of thinned, approximately independent MCMC samples from $[Y(\mathbf{s}_0)\mid \mathbf{Z}]$, the predictive distribution. The direct Monte Carlo estimator of \eqref{eqn:OPD_predictor} is,
\begin{equation}
\hat{\delta}_\lambda^*(\mathbf{Z};\mathbf{s}_0) = \begin{cases}
\left(M^{-1} \sum_{i=1}^M y_{(i)}^{\lambda + 1}\right)^{1/(\lambda + 1)},& \lambda \neq -1\\
\exp\!\left\{M^{-1} \sum_{i=1}^M \log(y_{(i)})\right\},& \lambda = -1.
\end{cases}\label{eqn:naive_estimator}
\end{equation}
The variance of \eqref{eqn:naive_estimator} can be large; the choice of $M$ to obtain an acceptably small variance depends on $\lambda$, as the following delta-method calculation shows:
\begin{align*}
&\mathrm{var}\!\left(\hat{\delta}_\lambda^*(\mathbf{Z};\mathbf{s}_0)\right)\\
&~~~~~~~\simeq \begin{cases}
M^{-1}(\lambda + 1)^{-2}\delta_\lambda^*(\mathbf{Z}; \mathbf{s}_0)^{-2\lambda}\mathrm{var}(Y(\mathbf{s}_0)^{\lambda + 1}\mid \mathbf{Z}), & \lambda \neq -1\\
M^{-1}\exp(2E(\log(Y(\mathbf{s}_0))\mid\mathbf{Z}))\mathrm{var}(\log(Y(\mathbf{s}_0))\mid \mathbf{Z}), & \lambda = -1.
\end{cases}
\end{align*}
The number of thinned MCMC samples $M$ and the magnitude of the predictor are factors that contribute to the precision of the estimated predictor. At this juncture, we point out there is the potential to use \textit{control variates} to obtain more precise Monte Carlo estimates. See \cite{Glasserman2003} for a general overview of control variates and \cite{Mira2013} for a promising way to select optimal control variates. We leave this topic for future research.

\subsubsection*{Delta-method approximation}

The OPD spatial predictor can be approximated via the delta method (see Section \ref{sec:delta_method_appendix} of the Online Supplement). Let $\mathrm{CV}(Y(\mathbf{s}_0)\mid\mathbf{Z}) \equiv \sqrt{\mathrm{var}(Y(\mathbf{s}_0)\mid \mathbf{Z})}/E(Y(\mathbf{s}_0)\mid \mathbf{Z})$ be the (conditional) coefficient of variation. The second-order delta-method approximation to the $(\lambda+1)$-th fractional moment of $[Y(\mathbf{s}_0) \mid \mathbf{Z}]$ is easily seen to be,
\begin{equation}
E(Y(\mathbf{s}_0)^{\lambda + 1} \mid \mathbf{Z}) \simeq E(Y(\mathbf{s}_0)\mid\mathbf{Z})^{\lambda + 1}\left(1 + \frac{\lambda(\lambda + 1)\mathrm{CV}(Y(\mathbf{s}_0)\mid \mathbf{Z})^2}{2}\right),\label{eqn:delta_method_fractional_moment}
\end{equation}
for $\lambda \neq -1$. This is simply the expectation with respect to $[Y(\mathbf{s}_0)\mid \mathbf{Z}]$ of the second-order Taylor series of the function $g(x) = x^{\lambda + 1},~x>0$, about $x = E(Y(\mathbf{s}_0)\mid\mathbf{Z})$. The approximation to \eqref{eqn:OPD_predictor} for $\lambda \neq - 1$ is completed by taking the second-order Taylor series of the function $h(u) = u^{1/(\lambda + 1)},~u>0,$ about $u = 1$. Hence, for $\lambda \neq -1$, further algebra yields
\begin{align}
&E(Y(\mathbf{s}_0)^{\lambda + 1} \mid \mathbf{Z})^{1/(\lambda + 1)}\nonumber\\
&~~~~~~\simeq E(Y(\mathbf{s}_0)\mid\mathbf{Z}) \left(1 + \frac{\lambda \mathrm{CV}(Y(\mathbf{s}_0)\mid \mathbf{Z})^2}{2} - \frac{\lambda^3 \mathrm{CV}(Y(\mathbf{s}_0)\mid \mathbf{Z})^4}{8}\right).\label{eqn:delta_method_opd_predictor}
\end{align}
That is, the OPD spatial predictor is approximately equal to the ubiquitous posterior mean multiplied by a factor that depends on $\lambda$ and the coefficient of variation of the predictive distribution $[Y(\mathbf{s}_0)\mid\mathbf{Z}]$. The quality of the approximation relies on the multiplicative factor being close to $1$, so it is expected to deteriorate when $|\lambda|$ is large. 

For $\lambda = -1$, a separate second-order delta-method approximation yields
\begin{align}
    &\exp\{E(\log(Y(\mathbf{s}_0))\mid\mathbf{Z})\} \nonumber\\ 
    &~~~~~~~~~~\simeq E(Y(\mathbf{s}_0)\mid\mathbf{Z}) \cdot \left(1 -\frac{\mathrm{CV}(Y(\mathbf{s}_0) \mid \mathbf{Z})^2}{2} + \frac{\mathrm{CV}(Y(\mathbf{s}_0) \mid \mathbf{Z})^4}{8} \right),\label{eqn:delta_method_opd_predictor_neg1}
\end{align}
which is simply \eqref{eqn:delta_method_opd_predictor} with $\lambda = -1$. 

\subsection{Block prediction and change-of-support}\label{sec:block_prediction}

Up to this point, we have only considered point-level prediction from point-level data. However, prediction of a block $B \subset D$ with area or volume $|B| > 0$ is an important topic in spatial statistics because useful spatial-inference questions are often asked at the level of blocks, which could refer to parcels of land, cubic metres of earth, etc. The problem of making block-level predictions from point-level data is a problem known as \textit{change-of-support} \cite[Ch. 5.2]{Cressie1993}.

\subsubsection*{Block-level OPD spatial prediction}

Consider the values of the process $Y(\cdot)$ at points within a block $B \subset D$, namely $\{Y(\mathbf{u}): \mathbf{u} \in B\}$. At each point, there is a loss function and an optimal predictor. This gives rise to the two spatial fields, $\{L_{PDL,\lambda}(\delta(\mathbf{Z}; \mathbf{u}), Y(\mathbf{u})): \mathbf{u} \in B\}$ and $\{\delta_\lambda^*(\mathbf{Z}; \mathbf{u}): \mathbf{u} \in B\}$. How can they be used to obtain an optimal predictor of $Y(B) = |B|^{-1} \int_B Y(\mathbf{u})~\mathrm{d}\mathbf{u}$, which we denote by $\delta_\lambda^*(\mathbf{Z}; B)$?

As with the equivalent point-level problem, the block-level optimal predictor is based on a loss function. Define the block-level loss function, 
\begin{align}
&L_{APDL,\lambda}\left(\left\{\delta(\mathbf{Z}; \mathbf{u}): \mathbf{u} \in B\right\}, \left\{Y(\mathbf{u}): \mathbf{u} \in B\right\}\right) \nonumber\\ &~~~~~~~~~~~~~~~~~~~\equiv|B|^{-1}\int_{B} L_{PDL,\lambda}(\delta(\mathbf{Z}; \mathbf{u}), Y(\mathbf{u}))~\mathrm{d}\mathbf{u},\label{eqn:block_level_PDL}
\end{align}
which is the average of the point-level loss functions within the block. Minimising this expected loss under the posterior distribution is analogous to the A-optimality criterion in experimental design \cite[e.g.,][]{Atkinson1992}. In this case, the ELP obtained from \eqref{eqn:block_level_PDL} is given by,
\begin{align}
    ELP_{\lambda}(\delta; \mathbf{Z}; B) &\equiv E\!\left(|B|^{-1}\int_{B} L_{PDL,\lambda}(\delta(\mathbf{Z}; \mathbf{u}), Y(\mathbf{u}))~\mathrm{d}\mathbf{u}\mid\mathbf{Z}\right)\nonumber \\
    &= |B|^{-1}\int_{B} ELP_\lambda(\delta; \mathbf{Z}; \mathbf{u})~\mathrm{d}\mathbf{u}.\label{eqn:block_level_ELP}
\end{align}
By extension, the ELJ is:
\begin{align}
    ELJ_\lambda(\delta; B) &= E\left(ELP_{\lambda}(\delta; \mathbf{Z}; B)\right)\nonumber\\ 
    &= |B|^{-1}\int_{B} E(ELP_{\lambda}(\delta;\mathbf{Z};\mathbf{u}))~\mathrm{d}\mathbf{u} = |B|^{-1} \int_{B} ELJ_\lambda(\delta;\mathbf{u})~\mathrm{d}\mathbf{u}.\label{eqn:block_level_ELJ}
\end{align}
Minimising \eqref{eqn:block_level_ELP} results in 
\begin{equation}
    \delta_\lambda^*(\mathbf{Z}; B) = |B|^{-1}\int_B \delta_\lambda^*(\mathbf{Z}; \mathbf{u})~\mathrm{d}\mathbf{u},\label{eqn:block_OPD_predictor}
\end{equation}
which is equal to the block-average of the point-wise OPD spatial predictors. Other summaries can be used to define block-level loss functions, such as the maximum of the point-level loss functions (analogous to the $M$-optimality criterion in experimental design).

\subsubsection*{Block prediction by a delta-method approximation}

Suppose that instead of \eqref{eqn:block_level_PDL}, the loss function for prediction of $Y(B)$ is $L_{PDL,\lambda}(\delta(\mathbf{Z};B), Y(B))$; in this case, the optimal predictor is $E(Y(B)^{\lambda + 1}\mid\mathbf{Z})^{1/(\lambda + 1)}$, which is different from \eqref{eqn:block_OPD_predictor}. Note that a different loss function usually (but not always) results in a different optimal predictor. For block $B$, the delta-method approximation to the OPD spatial predictor is obtained in a similar manner to \eqref{eqn:delta_method_opd_predictor}:
\begin{align*}
&E(Y(B)^{\lambda + 1} \mid \mathbf{Z})^{1/(\lambda + 1)} \simeq \\
&~~~~~~~~E(Y(B)\mid\mathbf{Z}) \left(1 + \frac{\lambda \mathrm{CV}(Y(B)\mid \mathbf{Z})^2}{2} - \frac{\lambda^3 \mathrm{CV}(Y(B)\mid \mathbf{Z})^4}{8}\right).
\end{align*}
This approximation allows for convenient block prediction since the block-predictive mean, $E(Y(B)\mid\mathbf{Z})$, and block-predictive variance, $\mathrm{var}(Y(B)\mid \mathbf{Z})$, used in the definition of $\mathrm{CV}(Y(B)\mid \mathbf{Z})$, can be computed from the corresponding means, variances, and covariances at point-level support.  

\subsection{Properties of the OPD spatial predictor}\label{sec:properties}

A key property of the OPD spatial predictor is that it is biased, except when $\lambda = 0$. The following subsections analyse various aspects of its biasedness.

\subsubsection*{Bias}

The bias of any spatial predictor $\delta(\mathbf{Z}; \mathbf{s}_0)$ is defined as,
\begin{equation}
Bias(\delta; \mathbf{s}_0) \equiv E\{\delta(\mathbf{Z}; \mathbf{s}_0) - Y(\mathbf{s}_0)\},\label{eqn:bias}
\end{equation}
where the expectation in \eqref{eqn:bias} is with respect to the joint distribution, $[Y(\mathbf{s}_0), \mathbf{Z}]$. What can be said about the bias of $\delta_\lambda^*(\mathbf{Z}; \mathbf{s}_0)$? By inspection of \eqref{eqn:OPD_predictor}, it is clear that $\lambda = 0$ yields an unbiased predictor since $\delta_0^*(\mathbf{Z};\mathbf{s}_0) = E(Y(\mathbf{s}_0)\mid\mathbf{Z})$. Jensen's inequality shows that $\delta_\lambda^*(\mathbf{Z};\mathbf{s}_0) > E(Y(\mathbf{s}_0)\mid \mathbf{Z})$ for $\lambda > 0$ and $\delta_\lambda^*(\mathbf{Z};\mathbf{s}_0) < E(Y(\mathbf{s}_0)\mid \mathbf{Z})$ for $\lambda < 0$ (see Section \ref{sec:bias_appendix} of the Online Supplement). In fact, $Bias(\delta_\lambda^*;\mathbf{s}_0)$ is a monotonic increasing function of $\lambda$, as we now show. 

Consider
\begin{align}
    \frac{\partial\delta_\lambda^*(\mathbf{Z}; \mathbf{s}_0)}{\partial \lambda} = \begin{cases}
    \frac{\left(E(Y(\mathbf{s}_0)^{\lambda + 1}\log(Y(\mathbf{s}_0)^{\lambda + 1}) \mid \mathbf{Z}) - E(Y(\mathbf{s}_0)^{\lambda + 1}\mid\mathbf{Z})\log(E(Y(\mathbf{s}_0)^{\lambda + 1}\mid\mathbf{Z}))\right)}{(\lambda + 1)^{2}\delta_\lambda^*(\mathbf{Z}; \mathbf{s}_0)^{\lambda}}, & \lambda \neq -1\\
    \frac{\mathrm{var}(\log(Y(\mathbf{s}_0))\mid\mathbf{Z})}{2\delta_{-1}^*(\mathbf{Z};\mathbf{s}_0)^{-1}},& \lambda = -1,
    \end{cases}\label{eqn:derivative_OPD}
\end{align}
and we now show that the derivative in \eqref{eqn:derivative_OPD} is always positive: When $\lambda = -1$, every term is a strictly positive quantity, and when $\lambda \neq -1$, $(\lambda + 1)^{-2}\delta_\lambda^*(\mathbf{Z};\mathbf{s}_0)^\lambda$ is strictly positive, and the bracketed term is positive by Jensen's inequality because the function $h(x) = x\log(x),~x>0,$ is convex. When predicting extreme parts of $Y(\cdot)$, bias can be desirable, and hence $\lambda$ is a `dial' that can be used to make conservative predictions.  

\subsubsection*{The minimised ELP and ELJ}

Under SEL, the bias of the optimal predictor, $E(Y(\mathbf{s}_0)\mid\mathbf{Z})$, is zero. Hence the minimised ELJ (i.e., the minimised mean-squared prediction error) of the optimal predictor is the expectation of the predictive variance, $E\{\mathrm{var}(Y(\mathbf{s}_0)\mid\mathbf{Z})\}$. A similarly interesting interpretation exists for the minimised ELP and ELJ of the OPD spatial predictor. It is straightforward to show that the minimised ELP, denoted as $ELP^*_\lambda(\delta_\lambda^*;\mathbf{Z};\mathbf{s}_0)$, is given by
\begin{align}
   ELP_{\lambda}^*(\delta_\lambda^*; \mathbf{Z}; \mathbf{s}_0) = \frac{\delta_\lambda^*(\mathbf{Z}; \mathbf{s}_0) - \delta_0^*(\mathbf{Z}; \mathbf{s}_0)}{\lambda},\label{eqn:minimised_ELP}
\end{align}
where recall that $\delta_0^*(\mathbf{Z}; \mathbf{s}_0) = E(Y(\mathbf{s}_0)\mid\mathbf{Z})$. For $\lambda = 0$, we define $ELP_{0}^*(\delta_0^*; \mathbf{Z}; \mathbf{s}_0)\equiv \lim_{\lambda \rightarrow 0} ELP_\lambda^*(\delta_\lambda^*; \mathbf{Z}; \mathbf{s}_0)$. After using L'H\^{o}pital's rule, 
\begin{align}
   ELP_{0}^*(\delta_0^*; \mathbf{Z}; \mathbf{s}_0) = E(Y(\mathbf{s}_0)\log(Y(\mathbf{s}_0))\mid\mathbf{Z}) - E(Y(\mathbf{s}_0)\mid \mathbf{Z})\log(E(Y(\mathbf{s}_0)\mid \mathbf{Z})).\label{eqn:minimised_ELP0}
\end{align}
Both \eqref{eqn:minimised_ELP} and \eqref{eqn:minimised_ELP0} are positive, the former because $\delta_\lambda^*(\mathbf{Z};\mathbf{s}_0) - \delta_0^*(\mathbf{Z}; \mathbf{s}_0)$ always has the same sign as $\lambda$ (see Section \ref{sec:bias_appendix} of the Online Supplement) and the latter because $g(x) = x \log(x), x > 0$ is a convex function, so Jensen's inequality guarantees that \eqref{eqn:minimised_ELP0} is positive. 

The minimised ELJ, denoted as $ELJ_\lambda^*(\delta_\lambda^*;\mathbf{s}_0)$, is simply the expectation of \eqref{eqn:minimised_ELP} taken over the data:
\begin{align}
    ELJ_\lambda^*(\delta_\lambda^*; \mathbf{s}_0) \equiv E(ELP_\lambda^*(\delta_\lambda^*; \mathbf{Z};\mathbf{s}_0)) = \frac{Bias(\delta_\lambda^*; \mathbf{s}_0)}{\lambda}.\label{eqn:minimised_ELJ}
\end{align}
Similarly, for $\lambda = 0$, 
\begin{align}
    ELJ_0^*(\delta_0^*; \mathbf{s}_0) \equiv E\left\{Y(\mathbf{s}_0)\log(Y(\mathbf{s}_0)) - E(Y(\mathbf{s}_0)\mid \mathbf{Z})\log(E(Y(\mathbf{s}_0)\mid \mathbf{Z})\right\}.\label{eqn:minimised_ELJ0}
\end{align}

It should be noted that the minimised ELP and ELJ are useful for uncertainty quantification, even though at first glance they are only functions of the first moment. Recall $Y(\cdot)$ is a positive spatial process, and positive processes invariably exhibit mean-variance relationships that allow terms like the bias to contain information about the predictive variability.

\section{Log-Gaussian spatial process}\label{sec:log_normal}

This section addresses the special case of a log-Gaussian spatial process. Suppose $Y(\cdot) \equiv \exp\{W(\cdot)\}$, where $W(\cdot)$ is a Gaussian process, and hence the process model denoted by $[Y(\cdot)]$, is a log-Gaussian spatial process. Write 
\begin{equation}
W(\cdot) \equiv \mathbf{x}(\cdot)'\bm\beta  + \eta(\cdot)+ \xi(\cdot),\label{eqn:w_process}
\end{equation} 
where $\mathbf{x}(\cdot)$ is a $p$-dimensional vector of covariates, and $\bm\beta$ is a $p$-dimensional vector of coefficients; these define a linear model that serves as the deterministic mean function for $W(\cdot)$. Further, $\mathrm{var}(\eta(\cdot)) = \sigma^2_\eta$ and $E(\eta(\cdot)) = -0.5\sigma^2_\eta$; similarly, $\mathrm{var}(\xi(\cdot)) = \sigma^2_{\xi}$ and $E(\xi(\cdot)) = -0.5\sigma^2_\xi$. 

The process $\eta(\cdot)$ is an $L_2$-continuous weakly stationary Gaussian process with mean $- 0.5\sigma^2_\eta$, and isotropic covariance function, $C_\eta(\lvert|\mathbf{h}\rvert|; \bm\theta) \equiv \mathrm{cov}(\eta(\mathbf{s}), \eta(\mathbf{s} + \mathbf{h}))$ for $\mathbf{h} \in \mathbb{R}^d$, that depends on covariance parameters $\bm\theta$; also define $\sigma^2_\eta \equiv C_\eta(0; \bm\theta)$. The notation for the argument of the covariance function is simplified by letting $h \equiv \lvert|\mathbf{h}\rvert| \geq 0$. The process $\xi(\cdot) + 0.5\sigma^2_\xi$ is a Gaussian white-noise process with mean $0$ and variance $\sigma_{\xi}^2$ that is also uncorrelated with $\eta(\cdot)$. It represents microscale variation in $W(\cdot)$ \cite[e.g, see][pp. 58-60]{Cressie1993}. Microscale variation is non-smooth variability present in the process at infinitesimal distances, which usually modelled (as we have done) by adding extra variability at the origin of the covariance function of a smooth process. Therefore, while $\eta(\cdot)$ has a continuous isotropic covariance function, $C_\eta(h; \bm\theta)$ for $h \geq 0$, the covariance function of $W(\cdot)$ is discontinuous at $h=0$. That is, $C_W(h; \bm\theta, \sigma^2_{\xi}) = \sigma^2_{\xi}\mathbb{I}(h = 0) + C_\eta(h; \bm\theta)$ for $h\geq 0$, where $\mathbb{I}(\cdot)$ is the indicator function. The marginal variance of the process is, $\sigma^2_{W} \equiv \mathrm{cov}(W(\mathbf{s}), W(\mathbf{s})) = C_W(0; \bm\theta, \sigma^2_{\xi}) = \sigma^2_\eta + \sigma^2_{\xi}$.

While the form of \eqref{eqn:w_process} is ubiquitous in spatial statistics, specifying $\eta(\cdot)$ and $\xi(\cdot)$ to have non-zero means is unusual. However, these offsets guarantee that, on the original scale, the marginal expectation of the process is $E(Y(\cdot)) = \exp\{\mathbf{x}(\cdot)'\bm\beta\}$. 

As for the distribution of the spatial data, $\mathbf{Z}$, recall that if $\nu(\cdot)$ is a spatially-varying positive function, then $\nu(\cdot) \cdot Y(\cdot)$ is a positive-valued spatial process. Here, for the data model, set $\nu(\cdot) = \exp\{\varepsilon(\cdot)\}$, where $\varepsilon(\cdot)$ is a Gaussian process of independent and identically distributed  Gaussian random variables representing the log-scale effect of measurement error. The result is that data $Z(\cdot) = Y(\cdot)\cdot \exp\{\varepsilon(\cdot)\}$, have a log-Gaussian distribution with multiplicative measurement error. Hence, each $Z(\mathbf{s}_i)$, for $~i=1,...,n,$ is modelled as,
\begin{equation}
    Z(\mathbf{s}_i) = Y(\mathbf{s}_i) \cdot \exp\{\varepsilon(\mathbf{s}_i)\} = \exp\{W(\mathbf{s}_i) + \varepsilon(\mathbf{s}_i)\}, \label{eqn:data_model}
\end{equation}
where $\varepsilon(\cdot)$ is uncorrelated with $\eta(\cdot)$ and $\xi(\cdot)$ and has variance $\sigma^2_\varepsilon$ and mean $-0.5\sigma^2_\varepsilon$. The log-scale measurement-error process $\varepsilon(\cdot)$ is chosen to have a non-zero mean, $-0.5\sigma^2_\varepsilon$, to ensure that $E(Z(\mathbf{s}_i)\mid Y(\mathbf{s}_i)) = Y(\mathbf{s}_i)$. Finally, under conditional independence, the  data model is $[\mathbf{Z} \mid Y(\cdot)] = \prod_{i=1}^n[Z(\mathbf{s}_i) \mid Y(\mathbf{s}_i)]$, where the $i$-th component in the product is log-Gaussian with log-scale mean $-0.5\sigma^2_\varepsilon$ and log-scale variance $\sigma^2_\varepsilon$.

In the EHM framework, for the log-Gaussian process, the predictive distribution for $Y(\mathbf{s}_0)$ can be obtained analytically by analysing the process and the data on the log-scale. Let $\tilde{\mathbf{Z}} \equiv (\tilde{Z}(\mathbf{s}_1), ..., \tilde{Z}(\mathbf{s}_n))'$ where, for $i=1,...,n$, $\tilde{Z}(\mathbf{s}_i) \equiv \log(Z(\mathbf{s}_i))$ follows a Gaussian distribution. Following \cite{CressieWikle2011}, the log-scale predictive distribution,  $[W(\mathbf{s}_0)\mid \tilde{\mathbf{Z}}]$, is Gaussian. First, define $\bm\Sigma_{W}$ to be the covariance matrix for $W(\cdot)$ at the observation locations, $\{\mathbf{s}_1, ..., \mathbf{s}_n\}$, with the $(i,j)$-th element given by $C_W(\lvert|\mathbf{s}_i - \mathbf{s}_j\rvert|;\bm\theta, \sigma^2_{\xi})$; alternatively, $\bm{\Sigma}_W = \sigma^2_{\xi}\mathbf{I}_n + \bm{\Sigma}_{\eta}$, where $\bm{\Sigma}_{\eta}$ is a covariance matrix with $(i,j)$-th element given by $C_\eta(\lvert|\mathbf{s}_i - \mathbf{s}_j\rvert|; \bm\theta)$ and $\mathbf{I}_n$ is the $n$-dimensional identity matrix. Then define, $\bm\Sigma_{\Tilde{Z}} \equiv \sigma^2_\varepsilon \mathbf{I}_n + \bm\Sigma_{W}$, to be the covariance matrix for $\Tilde{\mathbf{Z}}$. Further, let $\mathbf{c}(\mathbf{s}_0) \equiv (C_W(\lvert| \mathbf{s}_0 - \mathbf{s}_1 \rvert|; \bm\theta), ..., C_W(\lvert| \mathbf{s}_0 - \mathbf{s}_n \rvert|; \bm\theta))'$ be a vector of covariances between the prediction location, $\mathbf{s}_0$, and the observation locations, $\{\mathbf{s}_i: i = 1, ..., n\}$. Then, defining $\mathbf{X}$ to be a $n\times p$ matrix of covariates where the $i$-th row is $\mathbf{x}(\mathbf{s}_i)'$ for $i=1,...,n$, the mean of $[W(\mathbf{s}_0)\mid\Tilde{\mathbf{Z}}]$ is,
\begin{equation}
    E(W(\mathbf{s}_0)\mid\tilde{\mathbf{Z}}) = \mathbf{x}(\mathbf{s}_0)'\bm\beta - 0.5\sigma^2_W + \mathbf{c}(\mathbf{s}_0)'\bm\Sigma^{-1}_{\tilde{Z}}\left(\tilde{\mathbf{Z}} - E(\Tilde{\mathbf{Z}})\right)\!,\label{eqn:log_scale_predictive_mean}
\end{equation}
where recall $\sigma^2_W = \sigma^2_\eta + \sigma^2_\xi$, and $E(\Tilde{\mathbf{Z}}) = \mathbf{X}\bm\beta - 0.5(\sigma^2_W + \sigma^2_\varepsilon)\mathbf{1}_n$, with $\mathbf{1}_n$ being the $n$-dimensional vector of ones. Further, the log-scale predictive variance is,
\begin{equation}
    \mathrm{var}(W(\mathbf{s}_0)\mid \tilde{\mathbf{Z}}) = \sigma^2_W - \mathbf{c}(\mathbf{s}_0)'\bm\Sigma_{\tilde{Z}}^{-1}\mathbf{c}(\mathbf{s}_0).\label{eqn:log_scale_predictive_Var}
\end{equation}
Recognising that $[Y(\mathbf{s}_0)\mid\mathbf{Z}] = [\exp\{W(\mathbf{s}_0)\}\mid \tilde{\mathbf{Z}}]$, the predictive distribution on the original scale is log-Gaussian with parameters given by \eqref{eqn:log_scale_predictive_mean} and  \eqref{eqn:log_scale_predictive_Var}. 

\subsection{The OPD spatial predictor and its distribution}

Since $[Y(\mathbf{s}_0)\mid\mathbf{Z}]$ is log-Gaussian, the $(\lambda + 1)$-th fractional moment can be obtained analytically, as follows.
\begin{align*}
E(Y(\mathbf{s}_0)^{\lambda + 1}\mid\mathbf{Z}) &= E(\exp\{(\lambda + 1)W(\mathbf{s}_0)\}\mid\mathbf{Z})\\
&= \exp\!\left\{(\lambda + 1)E(W(\mathbf{s}_0)\mid \mathbf{Z}) + 0.5(\lambda + 1)^2\mathrm{var}(W(\mathbf{s}_0)\mid\mathbf{Z})\right\}.
\end{align*}
Therefore, the OPD spatial predictor given by \eqref{eqn:OPD_predictor} is,
\begin{equation}
    \delta_\lambda^*(\mathbf{Z};\mathbf{s}_0) = \exp\!\left\{E(W(\mathbf{s}_0)\mid \mathbf{Z}) + 0.5(\lambda + 1)\mathrm{var}(W(\mathbf{s}_0)\mid\mathbf{Z})\right\}.\label{eqn:log_normal_OPD_predictor}
\end{equation}
When $\lambda = -1$, the OPD spatial predictor is in fact the median of $[Y(\mathbf{s}_0)\mid\mathbf{Z}]$. Recall this case was treated separately in \eqref{eqn:OPD_predictor}, but here no separate derivation is needed. This is a further illustration that several different loss functions can result in the same predictor; the predictive median is also obtained from $L_{AEL}$ defined in Table \ref{table:common_loss_functions} (not only for log-Gaussian processes but more generally).

In addition to the predictive distribution being log-Gaussian, the predictor itself follows a log-Gaussian distribution, deriving its randomness from the data. From \eqref{eqn:log_scale_predictive_mean}-\eqref{eqn:log_normal_OPD_predictor},
\begin{align*}
\log(\delta_\lambda^*(\mathbf{Z};\mathbf{s}_0)) = \mathbf{x}(\mathbf{s}_0)'\bm\beta &+ 0.5\lambda\left(\sigma^2_W - \mathbf{c}(\mathbf{s}_0)' \bm\Sigma^{-1}_{\Tilde{Z}}\mathbf{c}(\mathbf{s}_0) \right) \\
&- 0.5\mathbf{c}(\mathbf{s}_0)' \bm\Sigma^{-1}_{\Tilde{Z}}\mathbf{c}(\mathbf{s}_0) + \mathbf{c}(\mathbf{s}_0)'\bm\Sigma_{\Tilde{Z}}^{-1}\left(\Tilde{\mathbf{Z}} - E(\Tilde{\mathbf{Z}})\right),
\end{align*}
where recall $\Tilde{\mathbf{Z}} = \log(\mathbf{Z})$, elementwise. Since $\mathbf{c}(\mathbf{s}_0)'\bm\Sigma_{\Tilde{Z}}^{-1}\left(\Tilde{\mathbf{Z}} - E(\Tilde{\mathbf{Z}})\right)$ has a Gaussian distribution with mean $0$ and variance $\mathbf{c}(\mathbf{s}_0)'\bm\Sigma_{\mathbf{Z}}^{-1}\mathbf{c}(\mathbf{s}_0)$, then $\log(\delta_\lambda^*(\mathbf{Z};\mathbf{s}_0))$ follows a Gaussian distribution with mean,
\begin{align}
 E\!\left\{\log(\delta_\lambda^*(\mathbf{Z}; \mathbf{s}_0))\right\} = \mathbf{x}(\mathbf{s}_0)'\bm\beta - 0.5\mathbf{c}(\mathbf{s}_0)' \bm\Sigma^{-1}_{\Tilde{Z}}\mathbf{c}(\mathbf{s}_0) + 0.5\lambda\left(\sigma^2_W - \mathbf{c}(\mathbf{s}_0)' \bm\Sigma^{-1}_{\Tilde{Z}}\mathbf{c}(\mathbf{s}_0) \right) ,\label{eqn:log_scale_mean_predictor}   
 \end{align}
 and variance,
 \begin{align}
\mathrm{var}\left\{\log(\delta_\lambda^*(\mathbf{Z}; \mathbf{s}_0))\right\} = \mathbf{c}(\mathbf{s}_0)'\bm\Sigma_{\Tilde{Z}}^{-1}\mathbf{c}(\mathbf{s}_0).\label{eqn:log_scale_variance_predictor}   
\end{align}
Hence $\delta_\lambda^*(\mathbf{Z}; \mathbf{s}_0)$ is log-Gaussian with parameters \eqref{eqn:log_scale_mean_predictor} and \eqref{eqn:log_scale_variance_predictor}. Therefore, the mean of $\delta_\lambda^*(\mathbf{Z};\mathbf{s}_0)$ is,
\begin{equation}
    E(\delta_\lambda^*(\mathbf{Z}; \mathbf{s}_0)) = \exp\!\left\{\mathbf{x}(\mathbf{s}_0)'\bm\beta + 0.5\lambda\left(\sigma^2_W - \mathbf{c}(\mathbf{s}_0)'\bm\Sigma^{-1}_{\Tilde{Z}}\mathbf{c}(\mathbf{s}_0)\right)\right\},\label{eqn:mean_log_normal_OPD_predictor}
\end{equation}
and its variance is,
\begin{align}
    \mathrm{var}(\delta_\lambda^*(\mathbf{Z}; \mathbf{s}_0)) &= \exp\!\left\{2\mathbf{x}(\mathbf{s}_0)'\bm\beta + \lambda\!\left(\sigma^2_W - \mathbf{c}(\mathbf{s}_0)'\bm\Sigma^{-1}_{\Tilde{Z}}\mathbf{c}(\mathbf{s}_0)\right)\right\} \left(\exp\!\left\{\mathbf{c}(\mathbf{s}_0)' \bm\Sigma^{-1}_{\Tilde{Z}}\mathbf{c}(\mathbf{s}_0)\right\}-1\right).\label{eqn:variance_log_normal_OPD_predictor}
\end{align}

\subsection{Statistical inference}\label{sec:lognormal_inference}

A spatial predictor of $Y(\mathbf{s}_0)$ has statistical properties that are important for inference. These are its bias, its MSPE, and the minimised ELP and ELJ, which, for the hierarchical log-Gaussian model given by \eqref{eqn:log_scale_predictive_mean} and \eqref{eqn:log_scale_predictive_Var}, can be derived from \eqref{eqn:mean_log_normal_OPD_predictor} and \eqref{eqn:variance_log_normal_OPD_predictor}.

Recall that bias is defined as, $Bias(\delta_\lambda^*; \mathbf{s}_0) \equiv E\{\delta_\lambda^*(\mathbf{Z};\mathbf{s}_0) - Y(\mathbf{s}_0)\},$ where the expectation is taken with respect to the joint distribution, $[Y(\mathbf{s}_0), \mathbf{Z}]$. Since $E(Y(\mathbf{s}_0)) = \exp\{\mathbf{x}(\mathbf{s}_0)'\bm\beta\}$, it follows that 
\begin{align}
    Bias(\delta_\lambda^*; \mathbf{s}_0) &= \exp\left\{\mathbf{x}(\mathbf{s}_0)'\bm\beta\right\}\left(\exp\left\{0.5\lambda\left(\sigma^2_W - \mathbf{c}(\mathbf{s}_0)'\bm\Sigma^{-1}_{\Tilde{Z}}\mathbf{c}(\mathbf{s}_0)\right)\right\}-1\right).\label{eqn:bias_lognormal}
\end{align} 
The variance of the predictor is \eqref{eqn:variance_log_normal_OPD_predictor}, but for prediction we need its MSPE for inference. Applying the definition of the MSPE, 
\begin{align*}
MSPE(\delta_\lambda^*; \mathbf{s}_0) &\equiv E((\delta_\lambda^*(\mathbf{Z}; \mathbf{s}_0) - Y(\mathbf{s}_0))^2) = \mathrm{var}(\delta_\lambda^*(\mathbf{Z}; \mathbf{s}_0) - Y(\mathbf{s}_0)) + Bias(\delta_\lambda^*(\mathbf{Z}; \mathbf{s}_0))^2\\
&= \mathrm{var}(\delta_\lambda^*(\mathbf{Z}; \mathbf{s}_0)) + \mathrm{var}(Y(\mathbf{s}_0)) - 2\mathrm{cov}(Y(\mathbf{s}_0), \delta_\lambda^*(\mathbf{Z}; \mathbf{s}_0)) + Bias(\delta_\lambda^*(\mathbf{Z}; \mathbf{s}_0))^2,
\end{align*}
where $\mathrm{var}(Y(\mathbf{s}_0)) = \exp\{2\mathbf{x}(\mathbf{s}_0)'\bm\beta\}(\exp\{\sigma^2_W\}-1)$; $\mathrm{var}(\delta_\lambda^*(\mathbf{Z}; \mathbf{s}_0))$ is given by \eqref{eqn:variance_log_normal_OPD_predictor}; $Bias(\delta_\lambda^*; \mathbf{s}_0)$ is given by \eqref{eqn:bias_lognormal}; and $\mathrm{cov}(Y(\mathbf{s}_0), \delta_\lambda^*(\mathbf{Z}; \mathbf{s}_0))$ can be obtained by some simple algebra for log-Gaussian random variables. From the definition of covariance,
$$
\mathrm{cov}(Y(\mathbf{s}_0), \delta_\lambda^*(\mathbf{Z}; \mathbf{s}_0)) \equiv E(\delta_\lambda^*(\mathbf{Z}; \mathbf{s}_0) \cdot Y(\mathbf{s}_0)) - E(\delta_\lambda^*(\mathbf{Z}; \mathbf{s}_0)) \cdot E(Y(\mathbf{s}_0)),
$$
where $E(Y(\mathbf{s}_0)) = \exp\{\mathbf{x}(\mathbf{s}_0)'\bm\beta\}$, and $E(\delta_\lambda^*(\mathbf{Z};\mathbf{s}_0))$ is given by \eqref{eqn:mean_log_normal_OPD_predictor}. The remaining term, the expectation $E(\delta_\lambda^*(\mathbf{Z}; \mathbf{s}_0) \cdot Y(\mathbf{s}_0))$, can be calculated by noting that $\delta_\lambda^*(\mathbf{Z}; \mathbf{s}_0) \cdot Y(\mathbf{s}_0) = \exp\{\log(\delta_\lambda^*(\mathbf{Z}; \mathbf{s}_0)) + W(\mathbf{s}_0)\}$, where $W(\mathbf{s}_0)$ is Gaussian by definition, and we previously established that $\log(\delta_\lambda^*(\mathbf{Z}; \mathbf{s}_0))$ follows a Gaussian distribution with parameters given by \eqref{eqn:log_scale_mean_predictor} and \eqref{eqn:log_scale_variance_predictor}. After straightforward algebra to compute the mean and variance of $\log(\delta_\lambda^*(\mathbf{Z}; \mathbf{s}_0)) + W(\mathbf{s}_0)$, we obtain,
$$
E(\delta_\lambda^*(\mathbf{Z}; \mathbf{s}_0) \cdot Y(\mathbf{s}_0)) = \exp\{2\mathbf{x}(\mathbf{s}_0)'\bm\beta + \mathbf{c}(\mathbf{s}_0)'\bm\Sigma_{\Tilde{Z}}^{-1}\mathbf{c}(\mathbf{s}_0) + 0.5\lambda(\sigma^2_W - \mathbf{c}(\mathbf{s}_0)'\bm\Sigma^{-1}_{\Tilde{Z}}\mathbf{c}(\mathbf{s}_0))\}.
$$
Then, 
\begin{align}
    &\mathrm{cov}(\delta_\lambda^*(\mathbf{Z}; \mathbf{s}_0), Y(\mathbf{s}_0)) \nonumber\\
    &~~~~~~~~= \exp\!\left\{2\mathbf{x}(\mathbf{s}_0)'\bm\beta + 0.5\lambda(\sigma^2_W - \mathbf{c}(\mathbf{s}_0)'\bm\Sigma^{-1}_{\Tilde{Z}}\mathbf{c}(\mathbf{s}_0))\right\}\!\left(\exp\left\{\mathbf{c}(\mathbf{s}_0)'\bm\Sigma_{\Tilde{Z}}^{-1}\mathbf{c}(\mathbf{s}_0)\right\}- 1\right).\label{eqn:covariance_predictor_with_predictand}
\end{align}
Finally, combining \eqref{eqn:variance_log_normal_OPD_predictor}, \eqref{eqn:bias_lognormal}, and \eqref{eqn:covariance_predictor_with_predictand} in the expression for MSPE, we obtain
\begin{align}
   &MSPE(\delta_\lambda^*; \mathbf{s}_0)\nonumber\\
   &=\exp\left\{2\mathbf{x}(\mathbf{s}_0)'\bm\beta\right\}\Big(\exp\{\sigma^2_W\} - 2\exp\!\left\{\mathbf{c}(\mathbf{s}_0)'\bm{\Sigma}_{\Tilde{Z}}^{-1}\mathbf{c}(\mathbf{s}_0)) + 0.5\lambda(\sigma^2_W - \mathbf{c}(\mathbf{s}_0)'\bm{\Sigma}_{\Tilde{Z}}^{-1}\mathbf{c}(\mathbf{s}_0))\right\}\nonumber\\
   &~~~~~~~~~~~~~~~~~~~~~~~~~~~~~~~~~~~~~~~~+ \exp\!\left\{\mathbf{c}(\mathbf{s}_0)'\bm{\Sigma}_{\Tilde{Z}}^{-1}\mathbf{c}(\mathbf{s}_0)) + \lambda(\sigma^2_W - \mathbf{c}(\mathbf{s}_0)'\bm{\Sigma}_{\Tilde{Z}}^{-1}\mathbf{c}(\mathbf{s}_0))\right\}\Big).\label{eqn:lognormal_MSPE}
\end{align}
For the minimised ELP, we obtain for $\lambda \neq 0$,
\begin{align}
    &ELP_\lambda^*(\delta_\lambda^*; \mathbf{Z};\mathbf{s}_0) = \lambda^{-1}\exp\left\{\mathbf{x}(\mathbf{s}_0)'\bm\beta  - 0.5\mathbf{c}(\mathbf{s}_0)'\bm\Sigma^{-1}_{\Tilde{Z}}\mathbf{c}(\mathbf{s}_0) + \mathbf{c}(\mathbf{s}_0)'\bm\Sigma_{\Tilde{Z}}^{-1}(\Tilde{\mathbf{Z}} - E(\Tilde{\mathbf{Z}}))\right\} \nonumber\\
    &~~~~~~~~~~~~~~~~~\times\left(\exp\left\{0.5\lambda\left(\sigma^2_W - \mathbf{c}(\mathbf{s}_0)'\bm\Sigma_{\Tilde{Z}}^{-1}\mathbf{c}(\mathbf{s}_0)\right)\right\} - 1\right)\label{eqn:lognormal_ELP}
\end{align}
and the minimised ELJ is,
\begin{equation}
    ELJ_\lambda^*(\delta_\lambda^*; \mathbf{s}_0)=\lambda^{-1}\exp\left\{\mathbf{x}(\mathbf{s}_0)'\bm\beta\right\}\left(\exp\left\{0.5\lambda\left(\sigma^2_W - \mathbf{c}(\mathbf{s}_0)'\bm\Sigma^{-1}_{\Tilde{Z}}\mathbf{c}(\mathbf{s}_0)\right)\right\}-1\right).\label{eqn:lognormal_ELJ}
\end{equation}
For $\lambda = 0$, the minimised ELP is given by taking the limit of \eqref{eqn:lognormal_ELP} as $\lambda \rightarrow 0$. Hence, we obtain,
\begin{align}
&ELP_0^*(\delta_0^*; \mathbf{Z}; \mathbf{s}_0)\label{eqn:ELP_lognormal_0}\\
&=0.5\left(\sigma^2_W - \mathbf{c}(\mathbf{s}_0)' \bm\Sigma^{-1}_{\Tilde{Z}}\mathbf{c}(\mathbf{s}_0)\right)\exp\!\left\{\mathbf{x}(\mathbf{s}_0)'\bm\beta - 0.5\mathbf{c}(\mathbf{s}_0)' \bm\Sigma^{-1}_{\Tilde{Z}}\mathbf{c}(\mathbf{s}_0) + \mathbf{c}(\mathbf{s}_0)' \bm\Sigma^{-1}_{\Tilde{Z}}\!\left(\Tilde{\mathbf{Z}} - E(\Tilde{\mathbf{Z}})\right)\right\}\!.\nonumber
\end{align}
The minimised ELJ is the expectation of \eqref{eqn:ELP_lognormal_0}. After recognising that $\log(ELP_0^*(\delta_0^*; \mathbf{Z}; \mathbf{s}_0))$ follows a Gaussian distribution with mean,
$$
E\{\log(ELP_0^*(\delta_0^*; \mathbf{Z}; \mathbf{s}_0))\} = \log(0.5) +\log(\sigma^2_W - \mathbf{c}(\mathbf{s}_0)' \bm\Sigma^{-1}_{\Tilde{Z}}\mathbf{c}(\mathbf{s}_0)) + \mathbf{x}(\mathbf{s}_0)'\bm\beta - 0.5\mathbf{c}(\mathbf{s}_0)' \bm\Sigma^{-1}_{\Tilde{Z}}\mathbf{c}(\mathbf{s}_0)$$ 
and variance,
$$
\mathrm{var}\{\log(ELP_0^*(\delta_0^*; \mathbf{Z}; \mathbf{s}_0))\} = \mathbf{c}(\mathbf{s}_0)' \bm\Sigma^{-1}_{\Tilde{Z}}\mathbf{c}(\mathbf{s}_0),
$$ 
we obtain,
\begin{equation}
ELJ_0^*(\delta_0^*; \mathbf{s}_0) = 0.5\left(\sigma^2_W - \mathbf{c}(\mathbf{s}_0)' \bm\Sigma^{-1}_{\Tilde{Z}}\mathbf{c}(\mathbf{s}_0)\right)\exp\!\left\{\mathbf{x}(\mathbf{s}_0)'\bm\beta\right\}.\label{eqn:ELJ_lognormal_0}
\end{equation}

\subsection{Calibrating the OPD predictor to target a quantile}\label{sec:calibration}

The OPD spatial predictor and its bias are monotone increasing functions of $\lambda$. Here, for the OPD spatial predictor of the log-Gaussian hierarchical model, we illustrate how $\lambda$ can be calibrated to predict the $q$-th quantile of $[Y(\mathbf{s}_0)\mid\mathbf{Z}]$, for $q \in (0, 1)$. Let $F^{-1}(q; \mathbf{Z}; \mathbf{s}_0)$ be the log-Gaussian quantile function,
\begin{equation}
F^{-1}(q; \mathbf{Z}; \mathbf{s}_0) = \exp\!\left\{E(W(\mathbf{s}_0)\mid\Tilde{\mathbf{Z}}) + \sqrt{\mathrm{var}(W(\mathbf{s}_0)\mid\Tilde{\mathbf{Z}})} \cdot \Phi^{-1}\left(q\right)\right\},\label{eqn:quantile_function_lognormal}
\end{equation}
where $\Phi^{-1}(\cdot)$ is the quantile function of the standard Gaussian distribution. Then, setting \eqref{eqn:quantile_function_lognormal} equal to $\delta_\lambda^*(\mathbf{Z};\mathbf{s}_0)$ given by \eqref{eqn:log_normal_OPD_predictor} and solving for $\lambda$ yields the calibrated value,
\begin{align}
    \lambda^*_q(\mathbf{Z};\mathbf{s}_0) &= \frac{2\Phi^{-1}\left(q \right)}{\sqrt{\mathrm{var}(W(\mathbf{s}_0)\mid\Tilde{\mathbf{Z}})}} - 1.\label{eqn:conditional_calibrated_lambda}
\end{align}
When the median is desired, $q = 0.5$ and \eqref{eqn:conditional_calibrated_lambda} is $\lambda_{0.5}^*(\mathbf{s}_0) =-1$, which aligns with the earlier remark in this section that $\delta_{-1}^*(\mathbf{Z}; \mathbf{s}_0) = \exp\{E(W(\mathbf{s}_0)\mid\Tilde{\mathbf{Z}})\}$ predicts the median of $[Y(\mathbf{s}_0)\mid\mathbf{Z}]$.

Notice that the calibrated value of $\lambda$ from \eqref{eqn:conditional_calibrated_lambda} depends only on $q$ and the log-scale predictive variance at the prediction location. Because $\mathrm{var}(W(\mathbf{s}_0)\mid\Tilde{\mathbf{Z}})$ generally depends on $\mathbf{s}_0$, this calibrated $\lambda$ varies according to prediction location. The result is an adaptive spatial map, which can also adapt to the user's specifying $q$ to depend on $\mathbf{s}_0$. 

In theory, it should be possible to calibrate $\lambda$ to a specific quantile of the predictive distribution whenever the OPD spatial predictor and the quantile function of the predictive distribution are both available in closed form. Otherwise, standard root-finding algorithms could be used to solve for $\lambda$ in the equation, $\delta_\lambda^*(\mathbf{Z};\mathbf{s}_0) - F^{-1}(q; \mathbf{Z}; \mathbf{s}_0) = 0$, where both the predictor and quantile may have to be estimated from samples from the predictive distribution (with sufficiently high precision).

\section{Prediction intervals}\label{sec:prediction_intervals}

Prediction uncertainty for spatial prediction of $Y(\mathbf{s}_0)$ can be expressed through a single value like MSPE, but it is also common to provide a prediction interval at a pre-specified level of $100\times(1-\alpha)\%$, where, for example, $\alpha = 0.10, ~0.05, ~0.025$, etc.  Here we derive prediction intervals from a generic loss function with no assumptions made about $[Y(\mathbf{s}_0), \mathbf{Z}]$. First consider the unconditional prediction interval. Let $K^{(\alpha)}(\mathbf{s}_0) > 0$ be a `cut-off' with the same units as the loss function, and let $\alpha \in (0, 1)$ be a given level. Then, the $100\times(1-\alpha)\%$ unconditional prediction interval is defined by the equation,
\begin{equation}
    \mathrm{Pr}(L(\delta(\mathbf{Z}; \mathbf{s}_0), Y(\mathbf{s}_0)) \leq K^{(\alpha)}(\mathbf{s}_0)) = 1-\alpha.\label{eqn:general_unconditional_prediction_interval}
\end{equation}
The underlying probability measure is that of $[Y(\mathbf{s}_0), \mathbf{Z}]$, and the $100\times(1-\alpha)\%$ unconditional prediction interval is the set, 
$$
\{Y(\mathbf{s}_0): L(\delta(\mathbf{Z}; \mathbf{s}_0), Y(\mathbf{s}_0)) \leq K^{(\alpha)}(\mathbf{s}_0)\}.
$$
A $100\times(1-\alpha)\%$ conditional prediction interval can be formed by instead defining the cut-off as $C^{(\alpha)}(\mathbf{Z};\mathbf{s}_0) > 0$ that solves the equation,
\begin{equation}
\mathrm{Pr}(L(\delta(\mathbf{Z}; \mathbf{s}_0), Y(\mathbf{s}_0)) \leq C^{(\alpha)}(\mathbf{Z}; \mathbf{s}_0) \mid \mathbf{Z}) = 1-\alpha,\label{eqn:general_conditional_prediction_interval}
\end{equation}
where the underlying probability measure is that of $[Y(\mathbf{s}_0)\mid\mathbf{Z}]$, and the $100\times(1-\alpha)\%$ conditional prediction interval is the set,
$$
\{Y(\mathbf{s}_0): L(\delta(\mathbf{Z}; \mathbf{s}_0), Y(\mathbf{s}_0)) \leq C^{(\alpha)}(\mathbf{Z};\mathbf{s}_0)\}.
$$
The distinction between conditional and unconditional prediction intervals is not often made. 

If we now assume that $[Y(\cdot)]$ and $[\mathbf{Z} \mid Y(\cdot)]$ are Gaussian, and optimal spatial prediction is performed under SEL, the unconditional and conditional prediction intervals are identical, as we now show. Recall the squared-error loss function (Table \ref{table:common_loss_functions}), and let $Y(\cdot)$ be a Gaussian spatial process with noisy Gaussian spatial data $\mathbf{Z}$ at locations $\mathbf{s}_1, ..., \mathbf{s}_n$, so that $[Z(\mathbf{s}_i)\mid Y(\mathbf{s}_i)]$ is Gaussian with mean $Y(\mathbf{s}_i)$ and variance $\sigma^2_\varepsilon$, for $i = 1, ..., n$. Then the predictive distribution, $[Y(\mathbf{s}_0)\mid\mathbf{Z}]$, is also Gaussian, and an unconditional prediction interval can be derived from the probability statement in \eqref{eqn:general_unconditional_prediction_interval}:
\begin{align*}
\mathrm{Pr}(L_{SEL}(\delta^*_{SEL}(\mathbf{Z}; \mathbf{s}_0), Y(\mathbf{s}_0)) \leq K^{(\alpha)}(\mathbf{s}_0)) &= \mathrm{Pr}((Y(\mathbf{s}_0)-E(Y(\mathbf{s}_0)\mid \mathbf{Z}))^2 \leq K^{(\alpha)}(\mathbf{s}_0)), \\
&= \mathrm{Pr}\!\left(\frac{(Y(\mathbf{s}_0)-E(Y(\mathbf{s}_0)\mid \mathbf{Z}))^2}{\mathrm{var}(Y(\mathbf{s}_0) \mid \mathbf{Z})} \leq \kappa^{(\alpha)}\right),\\
&= 1- \alpha,
\end{align*}
where $\kappa^{(\alpha)} = K^{(\alpha)}(\mathbf{s}_0)/\mathrm{var}(Y(\mathbf{s}_0) \mid\mathbf{Z})$. In this case, since $Y(\mathbf{s}_0)$ is Gaussian, $(Y(\mathbf{s}_0) - E(Y(\mathbf{s}_0) \mid\mathbf{Z}))^2/\mathrm{var}(Y(\mathbf{s}_0) \mid\mathbf{Z}) \sim \chi^2_1$, and the cut-off $\kappa^{(\alpha)}$ is simply the $(1-\alpha)$ quantile of $\chi_1^2$, the chi-squared distribution on one degree of freedom. Importantly, although $\mathrm{var}(Y(\mathbf{s}_0)\mid\mathbf{Z})$ is notated in a way that suggests it depends on the data, in reality it does not because of the Gaussian assumption. Hence, deriving a prediction interval from the unconditional probability statement in \eqref{eqn:general_unconditional_prediction_interval} and the conditional probability statement in \eqref{eqn:general_conditional_prediction_interval} will yield the same result. Now, for $\alpha = 0.05$ (a common choice), $\kappa^{(0.05)} \simeq (1.96)^2$. Therefore, the lower and upper bounds of the $95\%$ prediction interval for $Y(\mathbf{s}_0)$ under SEL are given by the well-known prediction-interval formula, 
$$
E(Y(\mathbf{s}_0)\mid\mathbf{Z}) \pm 1.96\sqrt{\mathrm{var}(Y(\mathbf{s}_0)\mid\mathbf{Z})}.
$$
However, if anything about this scenario changes, such as the choice of loss function, the conditional and unconditional prediction intervals will quite likely be different.

\subsection{Conditional and unconditional prediction intervals under PDL}\label{sec:cond_uncond_pred_ints}

In our spatial hierarchical statistical model, both the process model and data model are positive-valued processes (hence non-Gaussian), and we are using PDL instead of SEL. Therefore, we expect there will be a difference between conditional and unconditional $100\times(1-\alpha)\%$ prediction intervals. 

\cite{Cressie2022} presented the unconditional $100\times(1-\alpha)\%$ prediction interval from PDL, whose cut-off $K_\lambda^{(\alpha)}(\mathbf{s}_0) > 0$ is defined by the equation,
$$
\mathrm{Pr}(L_{PDL,\lambda}(\delta_\lambda^*(\mathbf{Z}; \mathbf{s}_0), Y(\mathbf{s}_0)) \leq K_\lambda^{(\alpha)}(\mathbf{s}_0)) = 1 - \alpha.
$$
In general, $K_\lambda^{(\alpha)}(\mathbf{s}_0)$ can be obtained by Monte Carlo simulation. Recall the underlying probability measure is $[Y(\mathbf{s}_0), \mathbf{Z}] = \int [\mathbf{Z}\mid\mathbf{Y}][Y(\mathbf{s}_0), \mathbf{Y}]~\mathrm{d}\mathbf{Y}$, for $\mathbf{s}_0 \not\in \{\mathbf{s}_1, ..., \mathbf{s}_n\}$. Let $\{(y_0^{(m)}, \mathbf{y}^{(m)}): m = 1, ..., M\}$ be a sample from $[Y(\mathbf{s}_0), \mathbf{Y}]$, and let $\mathbf{z}^{(m)}$ be a set of spatial data simulated from the data model $[\mathbf{Z} \mid \mathbf{y}^{(m)}]$, given the $m$-th simulated process values. Keeping only the $(n+1)$-tuples, $\{(y_0^{(m)}, \mathbf{z}^{(m)}): m = 1, ..., M\}$, is the computational equivalent of integrating out $\mathbf{Y}$. Then, losses $\{L_{PDL,\lambda}(\delta_\lambda^*(\mathbf{z}^{(m)}; \mathbf{s}_0), y_0^{(m)}): m = 1, ..., M\}$ can be computed from the samples, and $K_\lambda^{(\alpha)}(\mathbf{s}_0)$ can be approximated as the $(1-\alpha)$-th empirical quantile of the $M$ computed losses, resulting in the $100\times(1-\alpha)\%$ unconditional prediction interval:
\begin{footnotesize}
    \begin{equation}
    \left\{Y(\mathbf{s}_0): Y(\mathbf{s}_0)^{\lambda + 1} - (\lambda + 1)\delta_\lambda^*(\mathbf{Z}; \mathbf{s}_0)^\lambda Y(\mathbf{s}_0) - \lambda\delta_\lambda^*(\mathbf{Z};\mathbf{s}_0)^\lambda\!\left((\lambda + 1)K_\lambda^{(\alpha)}(\mathbf{s}_0) - \delta_\lambda^*(\mathbf{Z}; \mathbf{s}_0)\right) \leq 0\right\}.\label{eqn:unconditional_prediction_interval}
\end{equation}
\end{footnotesize}

\noindent This is a set on the non-negative real line that can be obtained numerically by straightforward root-finding; see below.

To obtain the $100\times(1-\alpha)\%$ conditional prediction interval from PDL, the conditional cut-off $C_\lambda^{(\alpha)}(\mathbf{Z}; \mathbf{s}_0) > 0$ is defined by the equation,
$$
\mathrm{Pr}(L_{PDL,\lambda}(\delta_\lambda^*(\mathbf{Z}; \mathbf{s}_0), Y(\mathbf{s}_0)) \leq C_\lambda^{(\alpha)}(\mathbf{Z}; \mathbf{s}_0)\mid\mathbf{Z}) = 1 - \alpha.
$$
The simulations required to obtain $C_\lambda^{(\alpha)}(\mathbf{Z};\mathbf{s}_0)$ are slightly different from those required to obtain $K_\lambda^{(\alpha)}(\mathbf{s}_0)$. The underlying probability measure is the predictive distribution, $[Y(\mathbf{s}_0)\mid\mathbf{Z}]$, from which $M$ samples $\{y_0^{(m)}: m = 1, ..., M\}$ are needed. In general, this will most likely require MCMC computation because $[Y(\mathbf{s}_0)\mid\mathbf{Z}] = \int [Y(\mathbf{s}_0), \mathbf{Y}\mid\mathbf{Z}]~\mathrm{d}\mathbf{Y}$ is not usually available in closed form. From the $M$ samples, compute the corresponding losses, $\{L_{PDL,\lambda}(\delta_\lambda^*(\mathbf{Z}; \mathbf{s}_0), y_0^{(m)}): m = 1, ..., M\}$ for a given $\lambda$, and approximate $C_\lambda^{(\alpha)}(\mathbf{Z};\mathbf{s}_0)$ as the $(1-\alpha)$-th empirical quantile of the $M$ computed losses. Then, the conditional $100\times(1-\alpha)\%$ prediction interval can be obtained numerically by replacing $K_\lambda^{(\alpha)}(\mathbf{s}_0)$ with $C_\lambda^{(\alpha)}(\mathbf{Z};\mathbf{s}_0)$ in  \eqref{eqn:unconditional_prediction_interval}.

When $\lambda + 1$ is a positive integer, the left-hand side of the inequality in \eqref{eqn:unconditional_prediction_interval} is a polynomial in $Y(\mathbf{s}_0)$, resulting in analytical solutions for the lower and upper bound of the $100\times(1-\alpha)\%$ prediction interval when $\lambda \in \{1, 2, 3\}$ (see Sections \ref{sec:appendix_lambda1}-\ref{sec:appendix_lambda3} of the Online Supplement).  However, for general $\lambda$, the equation
\begin{equation}
Y(\mathbf{s}_0)^{\lambda + 1} - (\lambda + 1)\delta_\lambda^*(\mathbf{Z}; \mathbf{s}_0)^\lambda Y(\mathbf{s}_0) - \lambda\delta_\lambda^*(\mathbf{Z};\mathbf{s}_0)^\lambda\!\left((\lambda + 1)K_\lambda^{(\alpha)} - \delta_\lambda^*(\mathbf{Z}; \mathbf{s}_0)\right) = 0,\label{eqn:solve_to_find_interval}
\end{equation}
can be solved numerically to retrieve the lower ($l_\lambda$) and upper ($u_\lambda$) bounds of the prediction interval defined by \eqref{eqn:unconditional_prediction_interval}. Section \ref{sec:existence_appendix} of the Online Supplement contains a proof of the existence and uniqueness of a solution to \eqref{eqn:solve_to_find_interval} in the interval $(\delta_\lambda^*(\mathbf{Z}; \mathbf{s}_0), \infty)$, which then guarantees the existence of a one-sided or two-sided $100\times (1-\alpha)\%$ prediction interval. Section \ref{sec:prediction_interval_algorithm_appendix} of the Online Supplement contains a simple algorithm for solving \eqref{eqn:solve_to_find_interval} and therefore for finding $(l_\lambda, u_\lambda)$. The intervals are easy to visualise by plotting $L_{PDL,\lambda}(\delta_\lambda^*(\mathbf{Z}; \mathbf{s}_0), y)$ as a function of $y$ over a suitably large domain and observing where it intersects with the horizontal line at $K_\lambda^{(\alpha)}(\mathbf{s}_0)$ (or $C_\lambda^{(\alpha)}(\mathbf{Z}; \mathbf{s}_0)$).

\subsection{Coverage and interval widths}\label{sec:coverage}

The quality of a prediction interval can be assessed in many ways. The obvious method is to check its coverage. In theory, the $100\times(1-\alpha)\%$ prediction interval should cover $100\times(1-\alpha)\%$ of realisations of $Y(\mathbf{s}_0)$. The prediction interval \eqref{eqn:unconditional_prediction_interval} should satisfy this condition by definition, but several factors, including model specification and accuracy of the approximation of $K^{(\alpha)}_\lambda(\mathbf{s}_0)$ and even of $\delta_\lambda^*(\mathbf{Z};\mathbf{s}_0)$, can cause a breakdown in coverage. 

In empirical settings, the coverage of $100\times(1-\alpha)\%$ prediction intervals is calculated, for example, in leave-one-out cross-validation (LOOCV) as follows: one data point $Z(\mathbf{s}_i)$ is dropped from the spatial data resulting in $\mathbf{Z}_{-i} \equiv (Z(\mathbf{s}_1), ..., Z(\mathbf{s}_{i-1}), Z(\mathbf{s}_{i+1}), ..., Z(\mathbf{s}_n))'$, and then the prediction interval $\mathcal{I}_\alpha(\mathbf{s}_i; \mathbf{Z}_{-i})$ is obtained from the predictive distribution $[Z(\mathbf{s}_i)\mid \mathbf{Z}_{-i}]$, where note that the prediction is of the datum $Z(\mathbf{s}_i)$ instead of the latent process value $Y(\mathbf{s}_i)$. The empirical LOOCV coverage is then calculated as $n^{-1}\sum_{i=1}^n \mathbb{I}(Z(\mathbf{s}_{i}) \in \mathcal{I}_{\alpha}(\mathbf{s}_{i}; \mathbf{Z}_{-i}))$, where recall $\mathbb{I}(\cdot)$ is the indicator function. 

In Section \ref{sec:Meuse}, the predictive distribution $[Y(\mathbf{s}_i)\mid \mathbf{Z}_{-i}]$ and prediction interval for $Y(\mathbf{s}_i)$ are used instead of $[Z(\mathbf{s}_i)\mid \mathbf{Z}_{-i}]$ and the prediction interval for $Z(\mathbf{s}_i)$. If the measurement-error variance in the data is large, the coverage of an interval for $Y(\mathbf{s}_i)$ will be less than the nominal level. However, in Section \ref{sec:Meuse}, the measurement-error variance in the data is small and the coverages are accurate. 

The width of a prediction interval and the interval score \cite[]{Gneiting2007} are useful summary statistics. For example, \cite{Cressie2022} used the widths of the $100\times(1-\alpha)\%$ unconditional prediction interval to select $\lambda$ for OPD spatial prediction. The interval score for a prediction interval is its width plus penalty terms when realised values of the predicted quantity fall outside of the prediction interval. Note that the interval score is designed for `symmetric' prediction intervals in the sense that the probability below the lower-bound of the interval is the same as the probability above the upper-bound. Consequently, the interval score is generally not appropriate for the prediction intervals derived above. This leaves an open question of how to derive OPD prediction intervals that have equal probability below and above the lower and upper bounds, respectively.

\section{Application to mapping zinc in the soil of a Meuse River floodplain}\label{sec:Meuse}

The OPD spatial-prediction methodology discussed in the previous sections is now illustrated with an application to a dataset of zinc measurements in the soil of a floodplain of the Meuse River (hereafter referred to as `the Meuse River data') collected by \cite{Rikken1993}. There are 155 geolocated measurements of zinc and other heavy metals in the soil of a floodplain of the Meuse River immediately west of the town of Stein in the Netherlands (Fig. \ref{fig:map}). The dataset also contains a regular grid of 3,103 spatial-prediction locations over the study area. Every observation location and prediction location has a set of covariates associated with it, including distance from the nearest point on the Meuse River, flooding frequency, and soil type. The Meuse River data is a well known example dataset in the R package, \texttt{sp} \cite[]{spbook, spRnews}. 

\begin{figure}[!ht]
    \centering
    \includegraphics[width=\textwidth]{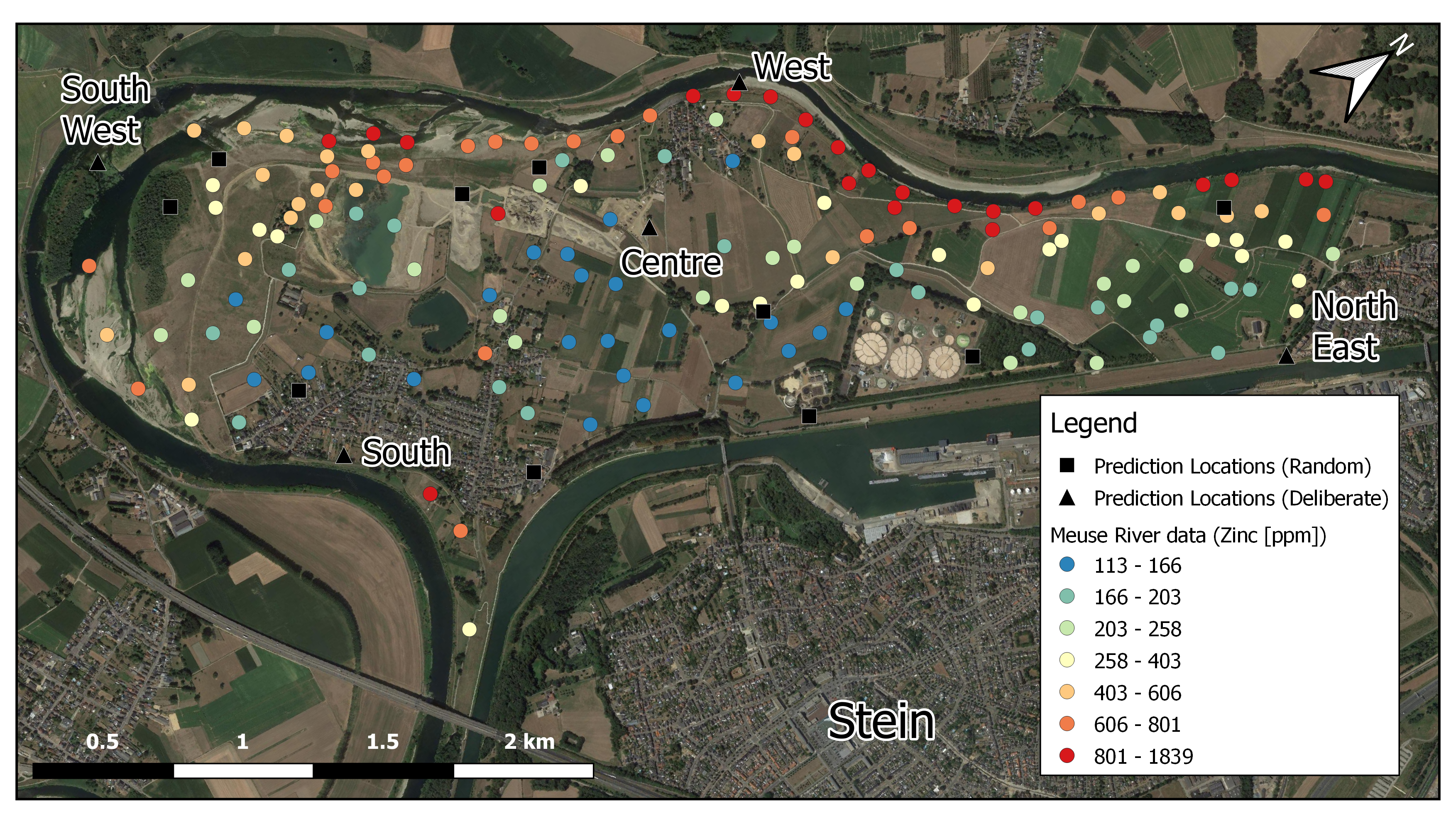}
    \caption{Map of the study area, west of the town of Stein, Netherlands. The arrow in the top-right corner points north. The round points are the observation locations, and the points are coloured according to the amount of zinc observed. A total of 15 prediction points are plotted on the map: 10 are randomly chosen prediction locations (black squares) and five are prediction locations specifically chosen to be on the boundary of the domain or, in one case, in the middle of the domain (black triangles). These 15 points are members of a much larger grid of prediction locations.}
    \label{fig:map}
\end{figure}

\cite{Cressie2022} chose 15 out of 3,103 prediction locations to use in special comparisons and analysis (e.g., to show curves of $ELJ_\lambda^*(\delta_\lambda^*; \mathbf{s}_0)$ and the widths of the 95\% unconditional prediction intervals, as functions of $\lambda$). The 15 points were chosen as follows: one was chosen to be in the middle of $D$ and four others were deliberately chosen near the boundaries. A further 10 points were selected at random. In this article, the same 15 points have been chosen for closer analysis (see Fig. \ref{fig:map} and its caption). Maps of the OPD spatial predictor and associated measures of uncertainty given in Section \ref{sec:meuse_examples} are plotted on the fine-resolution regular grid of all 3,103 prediction locations; see Section 6.3 below.

The rest of this section is organised as follows: Section \ref{sec:meuse_model} describes the spatial-statistical model for the Meuse River data on zinc concentrations (in ppm, parts per million). Then Section \ref{sec:estimation} describes how the parameters of the EHM in Section \ref{sec:meuse_model} were estimated. Finally, Section \ref{sec:meuse_examples} contains illustrations of the OPD-prediction methodology discussed in Sections \ref{sec:prediction}-\ref{sec:prediction_intervals}.

\subsection{Spatial statistical model}\label{sec:meuse_model}

The soil zinc measurements (in ppm) are modelled using the hierarchical statistical model outlined in Section \ref{sec:log_normal}. The locations $\mathbf{s}\in D$ are $(x,y)$-coordinate pairs of Eastings and Northings under a Netherlands map grid (European Petroleum Survey Group (EPSG) code 28992), and the spatial domain $D$ is the floodplain of the Meuse River as shown in Fig. \ref{fig:map}. The zinc measurements in the dataset are denoted $\mathbf{Z} \equiv (Z(\mathbf{s}_1), ..., Z(\mathbf{s}_{155}))'$, and the process $Y(\cdot) \equiv \{Y(\mathbf{s}): \mathbf{s} \in D\}$ represents the latent soil zinc concentrations (in ppm). Both take positive values.

The mean function of $Y(\cdot)$ was assumed to be a linear model with the same covariates used by \cite{Cressie2022}: distance to river, flooding frequency, soil type, and (standardised) x-coordinate. Flooding frequency and soil type are categorical covariates with three levels each. Therefore, $\mathbf{x}(\mathbf{s})$ and $\bm\beta$ are seven-dimensional vectors. 

An isotropic spherical covariance function was chosen to model the spatial dependency in $W(\cdot) \equiv \log(Y(\cdot))$. Let $h \equiv \lvert| \mathbf{h} \rvert|$ be the Euclidean distance between the two points $\mathbf{s} + \mathbf{h}$ and $ \mathbf{s} \in D$. Then, for the process $\eta(\cdot)$, which contains the smooth part of the spatial variation in $W(\cdot)$, we use the spherical covariance function,
\begin{equation}
C_\eta(h; \sigma^2_\eta, r) \equiv \begin{cases}
\sigma_\eta^2\left(1 - 1.5 h/r + 0.5(h/r)^3\right),& 0 < h \leq r\\
0, & h > r,
\end{cases}\label{eqn:spherical}
\end{equation}
where $r$ is the range parameter, and $\sigma^2_\eta \equiv C_\eta(0; \sigma^2_\eta, r)$ is the variance of $\eta(\mathbf{s})$ for any $\mathbf{s}\in D$. The process $W(\cdot)$ also contains a non-smooth microscale variation component that exists at infinitesimal separation distances, modelled by $\sigma_{\xi}^2$. The random vector $\bm\eta \equiv (\eta(\mathbf{s}_1), ..., \eta(\mathbf{s}_n))'$ has covariance matrix $\bm\Sigma_{\eta}$, where, for $h_{ij} \equiv \lvert|\mathbf{s}_i - \mathbf{s}_j\rvert|$, the $(i,j)$-th element is given by $C_\eta(h_{ij}; \sigma^2_\eta, r)$; and $\mathbf{W} \equiv (W(\mathbf{s}_1),...,W(\mathbf{s}_n))'$ has covariance matrix $\bm\Sigma_{W} = \bm\Sigma_\eta + \sigma_{\xi}^2 \mathbf{I}_n$, where $\mathbf{I}_n$ is the $n$-dimensional identity matrix. By extension, the spatial data $\Tilde{\mathbf{Z}} \equiv (\log(Z(\mathbf{s}_1)), ..., \log(Z(\mathbf{s}_n)))'$ has covariance matrix $\bm{\Sigma}_{\Tilde{Z}} = \bm\Sigma_{W} + \sigma^2_\varepsilon \mathbf{I}_n$. In general, $\sigma_{\xi}^2$
 and $\sigma_{\varepsilon}^2$ are hard to distinguish just using the data $\Tilde{\mathbf{Z}}$. They are both non-smooth components of variation that contribute additively to the marginal variance of the spatial data $\mathbf{Z}$, and they are not individually identifiable in the absence of replicated spatial measurements or further modelling assumptions. However, with the Meuse River data, we are in a position to apportion the variation correctly by using information from duplicated measurements \cite[][and see Section \ref{sec:estimation}]{Rikken1993}.
 
\cite{Cressie2022} used a power-normal (fourth-root-transformed) model to map the zinc concentrations, but the log-Gaussian model used here is valid and in line with the preposition of this article that power-transformed models are not Gaussian on the transformed scale (they are truncated Gaussian), and hence their inferences are only approximately valid. 

\subsection{Parameter estimation and `plug-in'}\label{sec:estimation}

The hierarchical model being used is an EHM, so fully Bayesian inference that includes inference on the parameters is not carried out. In an EHM, the parameters are estimated and then `plugged' into the model. The simplest way to do this for our model is to log-transform the zinc measurements and obtain parameter estimates on that scale; see the parameters that appear in $\eqref{eqn:w_process}$ and \eqref{eqn:data_model}. 

The first parameter to be estimated is the measurement-error variance. There are replicated measurements that are shown in a single figure of the Master's thesis where the Meuse River dataset originates \cite[][Fig. 11.1]{Rikken1993}. These allow for direct estimation of $\sigma^2_\varepsilon$. The exact values of the duplicate measurements were not published in the thesis, so the values were obtained directly from Fig. 11.1 using an image-analysis program \cite[ImageJ; ][]{ImageJ}.  \cite{Rikken1993} report there are 21 pairs of replicated measurements in total, but only 18 pairs could be distinguished during image analysis. The duplicate measurements are presented in Table \ref{table:duplicate_measurements}. Their spatial locations in $D$ are not known; as is seen below, these are not needed to estimate $\sigma^2_\varepsilon$.

\begin{table}[!ht]
    \centering

    \begin{tabular}{|c|c|c|}
    \hline
        Pair ($p$) & $Z_{p1}$ [ppm] & $Z_{p2}$ [ppm]\\
        \hline
        1 & 123 & 123 \\
        2 & 154 & 165 \\
        3& 175 & 165 \\
         4& 193 & 213 \\
         5&234 & 252 \\
         6&261 & 274 \\
         7&377 & 441 \\
         8&377 & 380 \\
         9&407 & 454 \\
         10&466 & 502 \\
         11&473 & 547 \\
         12&596 & 682 \\
         13&750 & 823 \\
         14&919 & 1155 \\
         15&1319 & 1463\\
         16&1542 & 1566\\
         17&1572 & 1524\\
         18&1671 & 1842\\
         \hline
    \end{tabular}
        \caption{Replicated measurements of soil zinc concentrations (in ppm) extracted from Fig. 11.1 in \cite{Rikken1993}, rounded to the nearest integer.}\label{table:duplicate_measurements}
\end{table}
Call the duplicates in Table \ref{table:duplicate_measurements},  $\{(Z_{p1}, Z_{p2}): p = 1, ..., P\}$, where here $P = 18$. Then from our spatial model,  
$$
\log(Z_{p1}) - \log(Z_{p2}) = \varepsilon_{p1} - \varepsilon_{p2},
$$
where $\{\varepsilon_{pj}: j = 1, 2\}$ are independent and identically distributed  Gaussian measurement errors with mean $-0.5\sigma^2_\varepsilon$ and variance $\sigma^2_\varepsilon$. Hence, $\log(Z_{p1}) - \log(Z_{p2})$ is Gaussian distributed with mean $0$ and variance $2\sigma^2_\varepsilon$, which does not involve the non-zero mean or spatial location of the duplicate measurements. The maximum-likelihood estimator for $\sigma_\varepsilon^2$ is,
\begin{align*}
\hat{\sigma}^2_\varepsilon &= \frac{1}{2P}\sum_{p=1}^{P}(\log(Z_{p1}) - \log(Z_{p2}))^2.
\end{align*}
This estimator is unbiased, efficient, and consistent. We calculated $\hat{\sigma}^2_\varepsilon = 0.0053$ (ppm$^2$), which indicates that only a small amount of measurement error is present relative to the microscale variation, $\sigma_\xi^2$, in the process: The estimate $\hat{\sigma}^2_\xi$ is obtained below, from which the measurement error, as a percentage of the non-smooth variation present in the data, is $100\times\hat{\sigma}^2_\varepsilon/(\hat{\sigma}^2_\varepsilon + \hat{\sigma}^2_{\xi})\% = 14.42\%$. 

There are two other sets of parameters: the linear-model coefficients, $\bm\beta$, and the covariance parameters, $\bm\theta\equiv(\sigma^2_{\eta}, r, \sigma^2_{\xi})'$. In the standard geostatistical workflow for estimating these \cite[e.g., see][Section 5]{Cressie2022}, least-squares estimators and semivariograms are used to estimate first $\bm\beta$ and then $\bm\theta$. 

Here we augment the standard geostatistical parameter-estimation workflow to obtain a more efficient estimate of $\bm\beta$. First, let $\mathbf{X} \equiv (\mathbf{x}(\mathbf{s}_1), ..., \mathbf{x}(\mathbf{s}_{n}))'$ be the $n\times p$ matrix of covariates. Then, per \eqref{eqn:data_model}, the natural logarithm of the soil zinc concentrations can be modelled as $\Tilde{\mathbf{Z}} = \mathbf{X}\bm\beta + \bm\eta + \bm\xi + \bm\varepsilon$, where $\bm{\eta} \equiv (\eta(\mathbf{s}_1), ..., \eta(\mathbf{s}_n))'$, $\bm{\xi} \equiv (\xi(\mathbf{s}_1), ..., \xi(\mathbf{s}_n))'$, and $\bm\varepsilon \equiv (\varepsilon(\mathbf{s}_1), ..., \varepsilon(\mathbf{s}_n))'$. Recall that $\sigma^2_W = \sigma^2_\eta + \sigma^2_\xi$. Then the log-transformed spatial data $\Tilde{\mathbf{Z}}$ have mean, $E(\Tilde{\mathbf{Z}}) = \mathbf{X}\bm\beta - 0.5(\sigma^2_W + \sigma^2_{\varepsilon})\mathbf{1}_n$, where $\mathbf{1}_n$ is an $n$-dimensional vector of ones, and variance, $\mathrm{var}(\Tilde{\mathbf{Z}}) = \bm{\Sigma}_{\Tilde{Z}} = \bm{\Sigma}_{\eta} + (\sigma^2_\xi + \sigma^2_\varepsilon)\mathbf{I}_n$. By putting $-0.5\sigma^2_W$ and $-0.5\sigma_\varepsilon^2$ in the mean function and letting $\bm\eta^*$, $\bm\xi^*$, and $\bm\varepsilon^*$ be the mean-zero versions of their corresponding random vectors, $\bm\eta$, $\bm\xi$, and $\bm\varepsilon$, respectively, the model can then also be written as, 
$$
\Tilde{\mathbf{Z}} = \mathbf{X}\bm\beta - 0.5(\sigma_W^2 + \sigma_\varepsilon^2)\mathbf{1}_n + \bm{\eta}^* + \bm{\xi}^* + \bm{\epsilon}^*.
$$
The generalised least squares estimator of $\bm\beta$ is clearly,
\begin{equation}
 \hat{\bm\beta}_{GLS} \equiv (\mathbf{X}'\hat{\bm\Sigma}_{\Tilde{Z}}^{-1}\mathbf{X})^{-1}\mathbf{X}'\hat{\bm\Sigma}_{\Tilde{Z}}^{-1}\left(\Tilde{\mathbf{Z}} + 0.5(\sigma_W^2 + \sigma_\varepsilon^2)\mathbf{1}_n\right).\label{eqn:adjusted_GLS}   
\end{equation}
However, the estimate $\hat{\bm{\beta}}_{GLS}$ depends on covariance parameters that are unknown. Therefore, we used the following algorithm to estimate $\bm{\beta}$ and $\bm{\theta} \equiv (\sigma^2_\eta, r, \sigma^2_\xi)'$. (Recall that $\sigma^2_\varepsilon$ has already been estimated from duplicate measurements.) The iteration was started with an initial estimate of $\bm\beta$ given by ordinary least squares (OLS):
$$
\hat{\bm{\beta}}^{(0)} \equiv (\mathbf{X}'\mathbf{X})^{-1}\mathbf{X}'\Tilde{\mathbf{Z}}.
$$
Then the residuals, $\mathbf{e}^{(0)} \equiv \Tilde{\mathbf{Z}} - \mathbf{X}\hat{\bm\beta}^{(0)}$, are used to compute an empirical semivariogram using the Cressie-Hawkins robust estimator \cite[]{Cressie1980}. Subsequently, a spherical semivariogram model is fitted to the empirical semivariogram via weighted least squares (WLS) with ``Cressie weights'' \cite[]{Cressie1985}. The resulting estimates of the covariance parameters are labelled $\hat{\bm\theta}^{(0)} \equiv ((\hat{\sigma}^2_\eta)^{(0)}, \hat{r}^{(0)}, (\hat{\sigma}^2_{\xi})^{(0)})'$, and these estimates are plugged into $\mathrm{var}(\Tilde{\mathbf{Z}})$ to obtain the initial covariance matrix, $\bm{\Sigma}^{(0)}_{\Tilde{Z}}$. It should be noted that the estimates of the covariance parameters come from a semivariogram model that has parameters $\sigma^2_\eta$ (partial sill), $r$ (range), and $c_0$ (nugget effect). It is known that the nugget effect $c_0 = \sigma^2_{\xi} + \sigma^2_{\varepsilon}$ \cite[pp. 58-60]{Cressie1993}, so letting $\hat{\sigma}^2_\eta, \hat{r},$ and $\hat{c}_0$ be the estimated parameters, we obtain $\hat{\sigma}^2_{\xi} = \max\{\hat{c}_0 - \hat{\sigma}^2_\varepsilon, 0\}$, and hence $\hat{\sigma}^2_W = \hat{\sigma}^2_\eta + \hat{\sigma}^2_{\xi}$.

The iteration proceeds from these starting values as follows. At the $i$-th iteration, compute an updated estimate of the linear-model coefficients using \eqref{eqn:adjusted_GLS}:
$$
\hat{\bm\beta}_{GLS}^{(i)} = (\mathbf{X}'(\hat{\bm\Sigma}_{\Tilde{Z}}^{(i-1)})^{-1}\mathbf{X})^{-1}\mathbf{X}'(\hat{\bm\Sigma}^{(i-1)}_{\Tilde{Z}})^{-1}\left(\Tilde{\mathbf{Z}} + 0.5((\hat{\sigma}_W^2)^{(i-1)} + \hat{\sigma}_\varepsilon^2)\mathbf{1}_n\right).
$$
Then compute the updated residuals,
$$
\mathbf{e}^{(i)} \equiv \Tilde{\mathbf{Z}} - \mathbf{X}\hat{\bm\beta}^{(i)}_{GLS} + 0.5((\hat{\sigma}_W^2)^{(i-1)} + \hat{\sigma}_\varepsilon^2)\mathbf{1}_n,
$$
subsequently re-compute the empirical semivariogram on the new residuals, and re-fit the spherical semivariogram model. This yields the updated covariance parameters, $\hat{\bm\theta}^{(i)}$, and hence the updated covariance matrix, $\bm{\Sigma}^{(i)}_{\Tilde{Z}}$. The iteration is continued until the element-wise absolute difference between $\hat{\bm{\beta}}_{GLS}^{(i)}$ and $\hat{\bm{\beta}}_{GLS}^{(i-1)}$ falls below a given tolerance (we used $10^{-6}$), and the latest estimates of $\bm{\beta}$ and $\bm{\theta}$, denoted $\hat{\bm{\beta}}_{GLS}$ and $\hat{\bm{\theta}}$, respectively, are used to describe the spatial statistical model for $Y(\cdot)$. In this application where $p=7$ and $n = 155$, six iterations were required for convergence. 

The resulting estimates and 95\% confidence intervals of the linear-model coefficients are given in Table \ref{table:parameter_estimates}. Distance and flooding frequency appear to have statistically significant effects on the logarithm of soil zinc, while the x-coordinate and soil type do not. The final covariance-parameter estimates were $\hat{\sigma}^2_\eta = 0.1855$ ppm$^2$, $\hat{r} = 991.76$ m, $\hat{\sigma}^2_{\xi} = 0.0313$ ppm$^2$, and hence $\hat{\sigma}_W^2 = 0.2168$ ppm$^2$. The final empirical semivariogram and fitted spherical model are shown in Fig. \ref{fig:semivariogram}. 

Although we used a log-transformation of the data to estimate the parameters, this is not the same as analysing the data on the log-scale. Inference (i.e., prediction) is done on the natural scale of the data and process, which is the positive real line, so predictions from this model do not need to be un-transformed and have their statistical properties approximated, as is done in transgaussian kriging \cite[]{Cressie1993}. 

\begin{figure}[!ht]
    \centering
    \includegraphics[width=\textwidth]{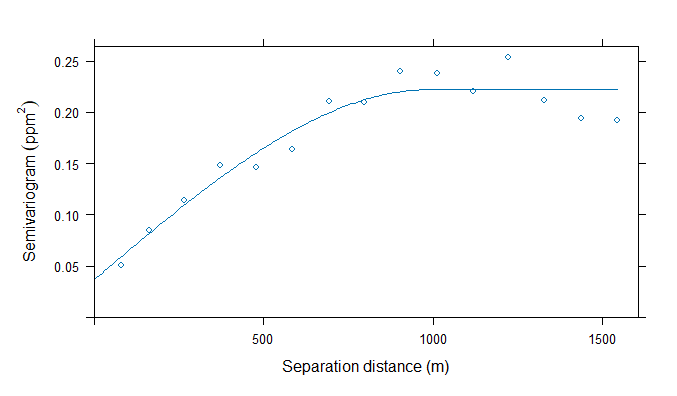}
    \caption{Empirical semivariogram (blue dots) and fitted spherical semivariogram model (blue line). This is the semivariogram fit obtained from parameters $\hat{\bm\theta}$ given at the final iteration of the parameter-estimation algorithm.}
    \label{fig:semivariogram}
\end{figure}

\begin{table}[]
    \centering
    \begin{tabular}{|l|l|ll|}
    \hline
    Parameter  & Description & Estimate & 95\% confidence interval\\
    \hline
    $\beta_0$  & Intercept  & 6.973 & (6.680,7.266) \\
    $\beta_1$  & Distance to River  & -2.152 & (-2.934, -1.369) \\
    $\beta_2$  & Soil Type 2  & -0.150 & (-0.354. 0.055) \\
    $\beta_3$  & Soil Type 3 & -0.130 & (-0.438, 0.178)\\
    $\beta_4$  & Level 2 Flood Frequency  & -0.593 & (-0.718, -0.468) \\
    $\beta_5$  & Level 3 Flood Frequency  & -0.621 & (-0.812, -0.431) \\
    $\beta_6$  & x-coordinate  & -0.148 & (-0.378, 0.083)\\
         \hline
    \end{tabular}
    \caption{Generalised least squares (GLS) estimates and 95\% confidence intervals of the linear-model coefficients in the log-Gaussian spatial process model.}
    \label{table:parameter_estimates}
\end{table}

\subsection{Spatial prediction and uncertainty quantification}\label{sec:meuse_examples}

Zinc is a heavy metal that can contaminate soil. Plants, invertebrate life, and small animals such as birds are particularly sensitive to it. In assessing the level of soil contamination by zinc (and other heavy metals), over-prediction of the concentration errs on the side of safety. The cost of over-prediction may be measured on the scale of tens of thousands of dollars if soil remediation is aggressively carried out on more soil than is necessary or if crops are not planted in a certain area due to fears of heavy-metal accumulation. While the cost of under-prediction is harder to quantify, it is almost certainly more severe: heavy metals can have long-term health consequences on people, reduce the yield of crops, or affect land valuations. In some cases, costly litigation may be brought, where the courts decide what the costs are. In what follows, we use the Meuse River dataset of soil zinc concentrations from \cite{Rikken1993} to illustrate spatial prediction and uncertainty quantification using power-divergence loss functions.

A key decision in any spatial statistical analysis using OPD spatial predictors is the choice of the power parameter $\lambda$. \cite{Cressie2022} proposed a simple statistical procedure to choose $\lambda$: that is, they computed the widths of the 95\% unconditional prediction intervals as functions of $\lambda$, at several prediction locations, and they considered the $\lambda$-value that minimised the width at each location. They settled on the median of these $\lambda$-values as the power for which OPD spatial prediction was carried out.

If two prediction intervals have the same coverage but one is narrower than the other, the narrower interval enables more precise inference on the random quantity of interest. However, there is a danger in using the widths of prediction intervals alone, since the width of a prediction interval says nothing of its coverage. In this sense, the interval score of \cite{Gneiting2007} combines both and indicates the true quality of a prediction interval. However, as discussed in Section \ref{sec:coverage}, our 95\% prediction intervals are not of the form required for interval scores. Instead, we checked the coverage of the intervals first, to ensure they were valid, and then we compared the widths.

The coverage of the 95\%  conditional and unconditional prediction intervals for the Meuse River data were calculated in a leave-one-out cross-validation (LOOCV) setting, using the algorithm defined in Section \ref{sec:coverage}. That is, each of the 155 zinc samples in the Meuse River data was predicted using the OPD spatial predictor as a function of the other 154 samples; then the proportion of times the left-out datum was contained in this prediction interval was calculated. Across the 155 sampling locations, the proportion should be approximately 95\%. Measurement error is not expected to strongly influence the outcome because we know the measurement error variance is small relative to the marginal variance of the data on the log-scale ($100\times \hat{\sigma}^2_\varepsilon/(\hat{\sigma}^2_W + \hat{\sigma}^2_\varepsilon)\% = 2.37\%$). To find the cut-offs needed to compute the intervals (see \eqref{eqn:general_unconditional_prediction_interval}
 and \eqref{eqn:general_conditional_prediction_interval}), the algorithms defined in Section \ref{sec:cond_uncond_pred_ints} were used. The number of Monte Carlo samples used to estimate the cut-offs was $M = 100,000$ samples from $[Y(\mathbf{s}_0)\mid\mathbf{Z}]$ or from $[Y(\mathbf{s}_0), \mathbf{Z}]$ for the  $95\%$ conditional and unconditional prediction intervals, respectively. Table \ref{table:coverages} confirms that the coverage of the 95\% prediction intervals is approximately correct, with the empirical average coverage achieved by both conditional and unconditional prediction intervals being between 92.3\% and 94.8\%. Having checked that the coverages are valid, we can look for the $\lambda$ that minimises the widths of the prediction intervals.

\begin{table}[!ht]
    \centering
    \begin{tabular}{|c|ccccccc|}
    \hline
        $\lambda$ & $-3$ & $-2$ & $-1$ & $0$ & $1$& $2$& $3$ \\
         \hline
        Conditional & 0.923 & 0.923 & 0.923 & 0.923 & 0.935 & 0.935 & 0.935 \\
        Unconditional & 0.935 & 0.942 & 0.942 & 0.948 & 0.948 & 0.948 & 0.948 \\
         \hline
    \end{tabular}
    \caption{Empirical leave-one-out cross-validation (LOOCV) coverages of 95\% conditional and unconditional prediction intervals for $\lambda \in \{-3, -2, -1, 0, 1, 2, 3\}$.}
    \label{table:coverages}
\end{table}

\begin{figure}
    \centering
    \includegraphics[width=\textwidth]{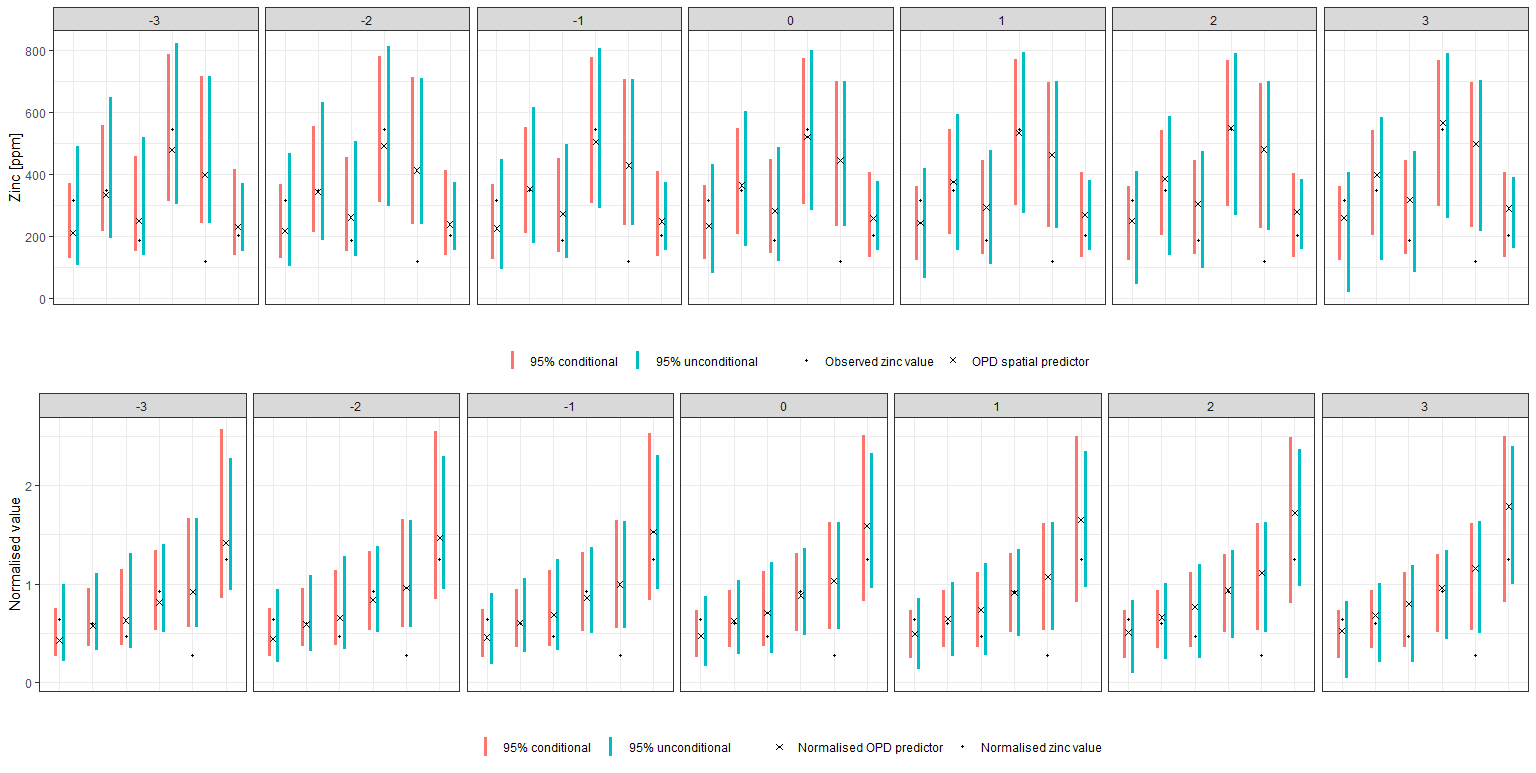}
    \caption{Top row: 95\% conditional and unconditional prediction intervals, $\delta_\lambda^*(\mathbf{Z};\mathbf{s}_i)$, and the observed zinc concentration at a small number (six) of observation locations for $\lambda \in \{-3, -2, -1, 0, 1, 2, 3\}$. Bottom row: Plots of quantities in the top row normalised by the marginal mean of the process at the selected locations, $E(Y(\mathbf{s}_i)) = \exp\{\mathbf{x}(\mathbf{s}_i)'\hat{\bm\beta}_{GLS}\}$. Here, only six of 155 observation locations are plotted due to lack of space. These were chosen by sorting the 155 observation locations according to the normalised values, $\delta_\lambda^*(\mathbf{Z}; \mathbf{s}_i)/\exp\{\mathbf{x}(\mathbf{s}_i)'\hat{\bm\beta}_{GLS}\}$, and selecting the first, 31st, 62nd, 93rd, 124th, and 155th locations from the ordered list.}
    \label{fig:coverage_check}
\end{figure}

Given that conditional and unconditional prediction intervals are different under power-divergence loss, it is natural to ask how they differ. These two forms of inference result in prediction intervals of different widths, as can be seen in Fig. \ref{fig:coverage_check}. The top row of plots shows unnormalised interval widths, but in the bottom row of plots in Fig. \ref{fig:coverage_check}, a clear pattern emerges when the interval widths (and OPD spatial predictors) are normalised by $E(Y(\mathbf{s}_i)) = \exp\{\mathbf{x}(\mathbf{s}_i)'\hat{\bm\beta}_{GLS}\}$, the estimated marginal mean of the process at the prediction location. This type of normalisation is done when using the coefficient of variation as a summary statistic. 

Now consider Fig. \ref{fig:widths_analysis}: When the ratios of the widths of the 95\% conditional and unconditional prediction intervals are plotted against the normalised OPD spatial predictor, namely $\delta_\lambda^*(\mathbf{Z};\mathbf{s}_0)/\exp\{\mathbf{x}(\mathbf{s}_0)'\hat{\bm\beta}_{GLS}\}$, for the 3,103 prediction locations, the relationship between the conditional and unconditional prediction intervals is abundantly clear. Fig. \ref{fig:widths_analysis} shows that, when $\delta_\lambda^*(\mathbf{Z};\mathbf{s}_0)/\exp\{\mathbf{x}(\mathbf{s}_0)'\hat{\bm\beta}_{GLS}\}$ is small, the conditional prediction intervals are narrower (ratio of the widths $< 1$) than their unconditional counterparts; when $\delta_\lambda^*(\mathbf{Z};\mathbf{s}_0)/\exp\{\mathbf{x}(\mathbf{s}_0)'\hat{\bm\beta}_{GLS}\}$ is large, the conditional prediction intervals are wider (ratio of the widths $> 1$) than their unconditional counterparts. Where this changeover occurs from `small' to `large' varies with $\lambda$, approximately, equal to $0.85$ ($\lambda = -3$), $0.9$ ($\lambda = -2$), $0.95$ ($\lambda = -1$), $1$ ($\lambda = 0$), $1.1$ ($\lambda = 1$), $1.15$ ($\lambda = 2$), $1.2$ ($\lambda = 3$).

\begin{figure}
    \centering
    \includegraphics[width=\textwidth]{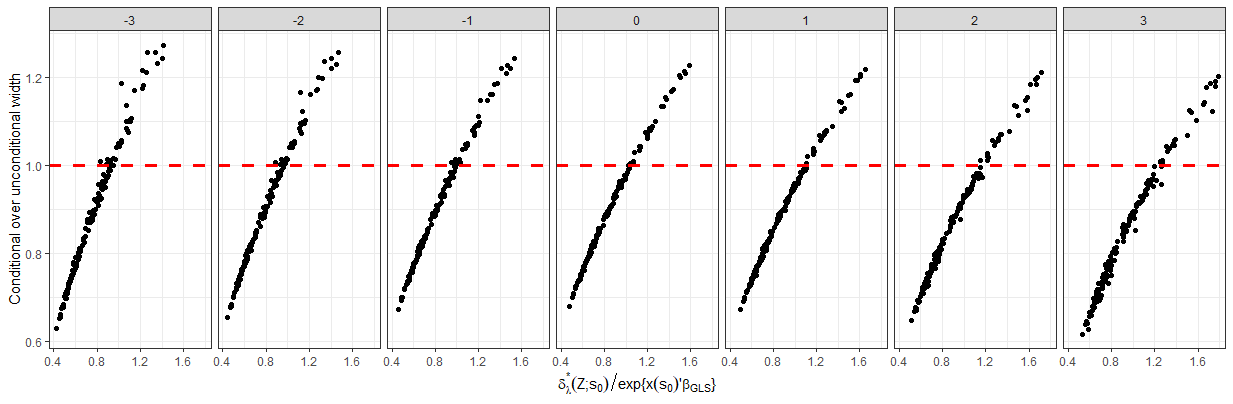}
    \caption{The relationship between $\delta_\lambda^*(\mathbf{Z};\mathbf{s}_{0})/\exp\{\mathbf{x}(\mathbf{s}_{0})'\hat{\bm\beta}_{GLS}\}$, where $\mathbf{s}_0$ ranges over the 3,103 prediction locations in $D$, and the ratio of the width of the 95\% conditional prediction intervals to the width of the 95\% unconditional prediction intervals for $\lambda \in \{-3, -2, -1, 0, 1, 2, 3\}.$ The dashed red line drawn across the vertical axis at $y = 1$ shows where the conditional prediction interval becomes wider than the unconditional prediction interval.}
    \label{fig:widths_analysis}
\end{figure}

Returning to the OPD spatial prediction of zinc concentrations, Fig. \ref{fig:PI_widths} shows the curves of the widths of the 95\% unconditional prediction intervals varying over $\lambda \in [-3, 3]$ for the five specially selected prediction locations shown in Fig. \ref{fig:map}. Let $\lambda^*(\mathbf{s}_0)$ be the value of $\lambda$ that minimises the width of the 95\% unconditional prediction interval at prediction location $\mathbf{s}_0 \in D$. Fig. \ref{fig:PI_widths} shows the histogram of $\lambda^*(\mathbf{s}_0)$ at all 15 prediction locations plotted in Fig. \ref{fig:map}. The median of these values is $\lambda^* = -0.5$; henceforth, our analysis of soil zinc concentrations uses the OPD spatial predictor,  
$$\delta_{-0.5}^*(\mathbf{Z}; \mathbf{s}_0) = E(Y(\mathbf{s}_0)^{1/2}\mid\mathbf{Z})^{2}.$$

\begin{figure}
    \centering
    \includegraphics[width=\textwidth]{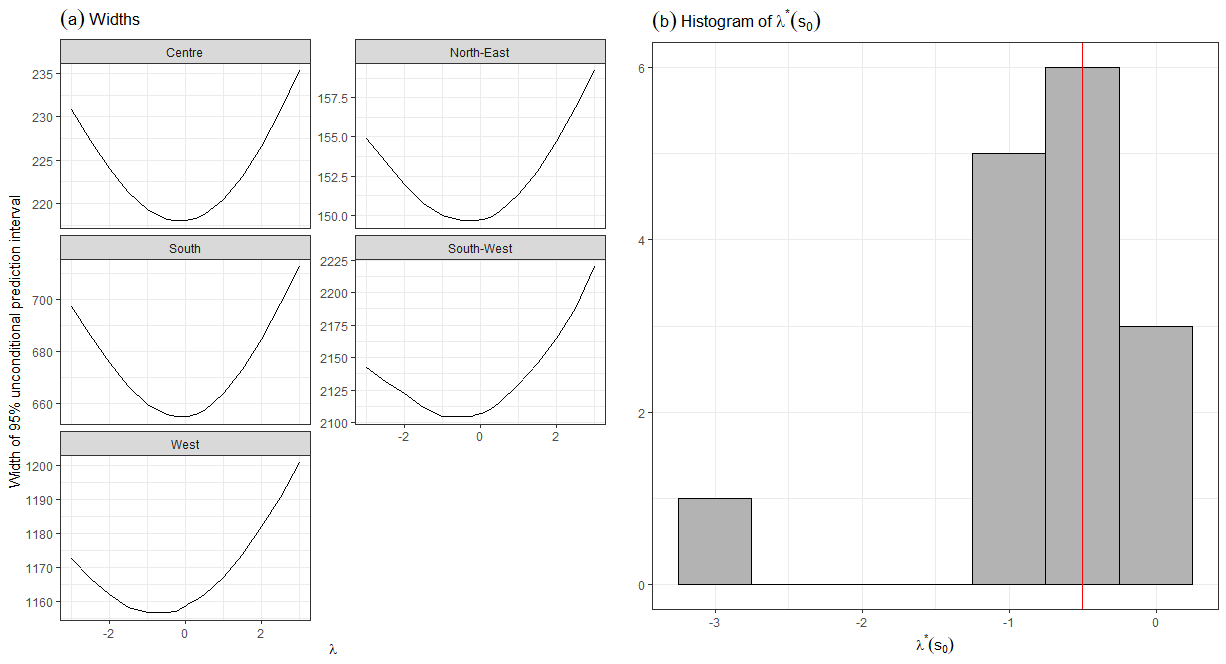}
    \caption{(a) Widths of 95\% unconditional prediction intervals as functions of $\lambda$ at the five specially selected prediction locations labelled in Fig. \ref{fig:map}. (b) Histogram of values $\lambda^*$ that minimise the widths of the 95\% unconditional prediction intervals at the 15 prediction locations shown in Fig. \ref{fig:map}.}
    \label{fig:PI_widths}
\end{figure}

The final step in the analysis is to use $\lambda = -0.5$ for prediction and uncertainty quantification. We computed $\delta_{-0.5}^*(\mathbf{Z}; \mathbf{s}_0)$, $Bias(\delta_{-0.5}^*; \mathbf{s}_0)$, $RMSPE(\delta_{-0.5}^*; \mathbf{s}_0) \equiv \sqrt{MSPE(\delta_{-0.5}^*;\mathbf{s}_0)}$, $ELP_{-0.5}^*(\delta_{-0.5}^*; \mathbf{Z}; \mathbf{s}_0)$, and $ELJ_{-0.5}^*(\delta_{-0.5}^*; \mathbf{s}_0)$ at all 3,103 prediction locations in $D$. Fig. \ref{fig:analysis} maps these quantities after they have been normalised by $E(Y(\mathbf{s}_0)) = \exp\{\mathbf{x}(\mathbf{s}_0)'\hat{\bm{\beta}}_{GLS}\}$, as was done in Fig. \ref{fig:coverage_check} and Fig. \ref{fig:widths_analysis}. Recall that scaling is a canonical operation for positive-valued spatial processes. From  \eqref{eqn:bias_lognormal}, \eqref{eqn:lognormal_MSPE}, \eqref{eqn:lognormal_ELP}, and \eqref{eqn:lognormal_ELJ} it is also clear that the normalised quantities shown in Fig. \ref{fig:analysis}(b)-\ref{fig:analysis}(e) do not depend on the mean-structure (i.e., they are stationary in space) and are unitless.

\begin{figure}[!h]
    \centering
    \includegraphics[width=\textwidth]{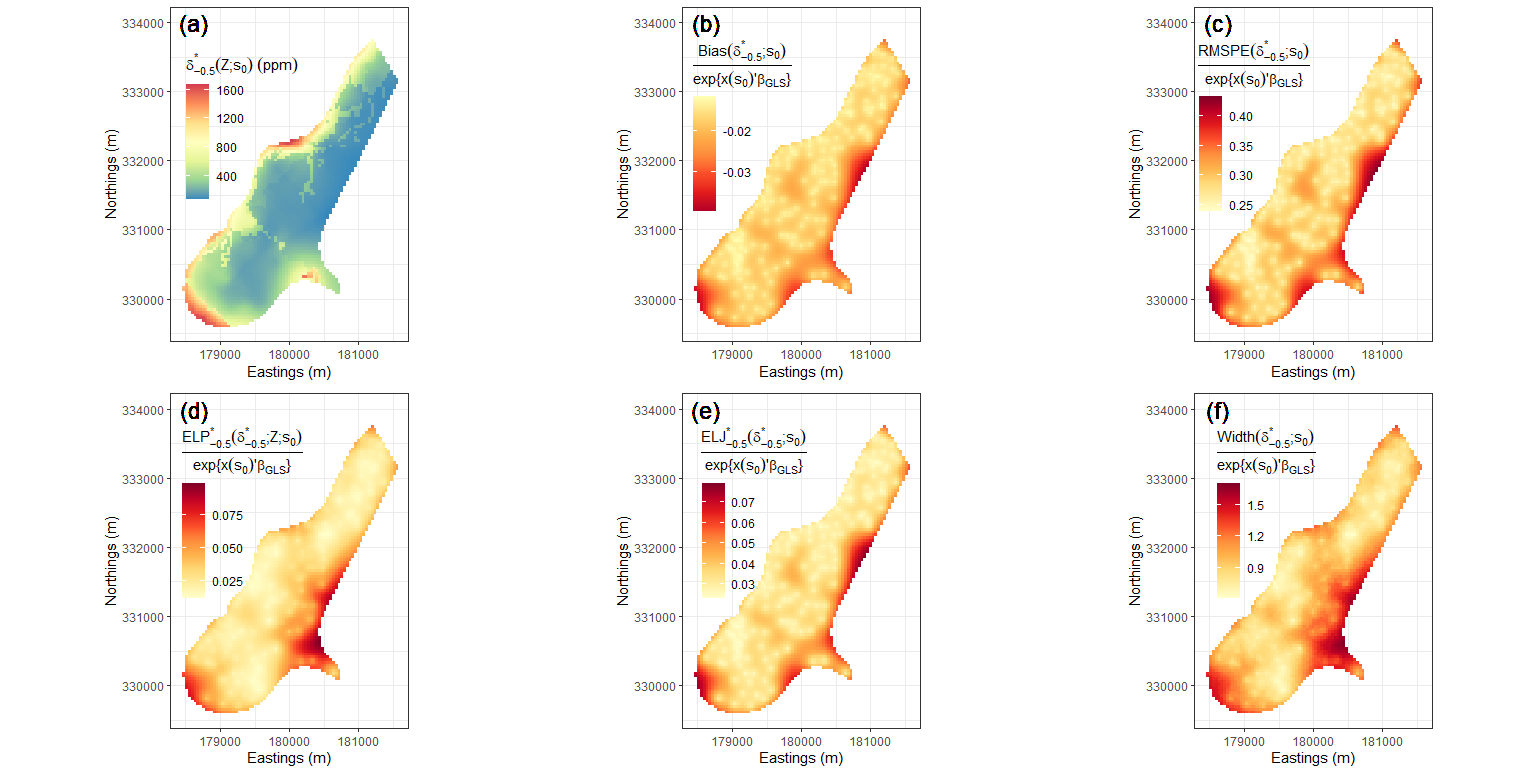}
    \caption{Maps of quantities: (a) $\delta_{-0.5}^*(\mathbf{Z}; \mathbf{s}_0)$, (b) $Bias(\delta_{-0.5}^*; \mathbf{s}_0)/\exp\{\mathbf{x}(\mathbf{s}_0)\hat{\bm\beta}_{GLS}\}$, (c) $RMSPE(\delta_{-0.5}^*; \mathbf{s}_0)/\exp\{\mathbf{x}(\mathbf{s}_0)'\hat{\bm\beta}_{GLS}\}$, (d) $ELP_{-0.5}^*(\delta_{-0.5}^*; \mathbf{Z}; \mathbf{s}_0)/\exp\{\mathbf{x}(\mathbf{s}_0)'\hat{\bm\beta}_{GLS}\}$, (e) $ELJ_{-0.5}^*(\delta_{-0.5}^*; \mathbf{s}_0)/\exp\{\mathbf{x}(\mathbf{s}_0)'\hat{\bm\beta}_{GLS}\}$, and (f) the widths of the 95\% unconditional prediction interval divided by $\exp\{\mathbf{x}(\mathbf{s}_0)'\hat{\bm\beta}_{GLS}\}$, where $\mathbf{s}_0$ ranges over $D$.}
    \label{fig:analysis}
\end{figure}

Under the log-Gaussian model assumed in Section \ref{sec:calibration}, another method for selecting $\lambda$ was suggested. At every prediction location $\mathbf{s}_0 \in D$, a different value of $\lambda^*_q(\mathbf{s}_0)$ can be computed using \eqref{eqn:conditional_calibrated_lambda} to predict the $q$-th quantile of $[Y(\mathbf{s}_0)\mid\mathbf{Z}]$. The value of $q$ is chosen \textit{a priori}; here, as an example, $q = 0.9$ was chosen to capture a scenario where we are predicting a `near-worst-case scenario' for the level of zinc contamination in the soil of the Meuse River floodplain. Fig. \ref{fig:adaptive_map_quantile}a shows $\delta_{\lambda_{0.9}^*(\mathbf{s}_0)}(\mathbf{Z}; \mathbf{s}_0)$, and Fig. \ref{fig:adaptive_map_quantile}b shows $\lambda_{0.9}^*(\mathbf{s}_0)$ over the study area. Compared to $\delta_{-0.5}^*(\mathbf{Z};\mathbf{s}_0)$ in Fig. \ref{fig:analysis}, the predictions $\delta_{\lambda_{0.9}^*(\mathbf{s}_0)}^*(\mathbf{Z};\mathbf{s}_0)$ in Fig. \ref{fig:adaptive_map_quantile}a are substantially higher, because of the calibration of the predictor to quite an extreme quantile. Consequently, the calibrated $\lambda$'s are in the range of 5-10 to capture the asymmetry due to the choice of $q = 0.9$. Also, recall from \eqref{eqn:conditional_calibrated_lambda} that $\lambda_{q}^*(\mathbf{s}_0)$ is inversely proportional to the log-scale predictive standard deviation. This is why Fig. \ref{fig:adaptive_map_quantile}(b) shows peaks around the observation locations. 

\begin{figure}[!h]
    \centering
    \includegraphics[width=\textwidth]{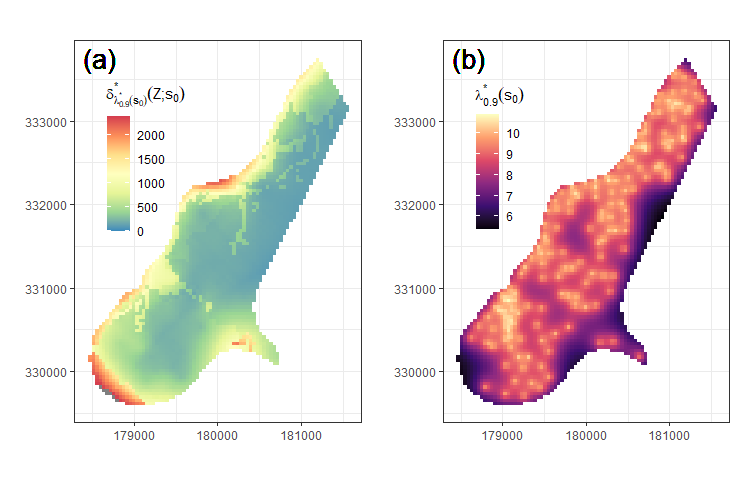}
    \caption{Maps of (a) the optimal power divergence (OPD) spatial predictor $\delta_{\lambda^*_{0.9}(\mathbf{s}_0)}^*(\mathbf{Z}; \mathbf{s}_0)$ that is  calibrated to predict the $0.9$ quantiles of the predictive distribution, $[Y(\mathbf{s}_0)\mid\mathbf{Z}]$, and of (b) the calibrated values, $\lambda^*_{0.9}(\mathbf{s}_0)$, where $\mathbf{s}_0$ ranges over $D$.}
    \label{fig:adaptive_map_quantile}
\end{figure}

For completeness, the maps of the OPD predictor from Fig. \ref{fig:adaptive_map_quantile}(a), and its normalised bias, RMSPE, ELP, ELJ, and the widths of the 95\% unconditional prediction intervals are shown in Fig. \ref{fig:adaptive_analysis}. It is interesting to see that the behaviour of RMSPE in Fig. \ref{fig:adaptive_analysis}c is radically different to that of RMSPE in Fig. \ref{fig:analysis}c. This is the result of the power parameter $\lambda$, of the loss function $L_{PDL,\lambda}$, being allowed to vary by location in Fig. \ref{fig:adaptive_analysis}, in contrast to the static choice of $\lambda$ in Fig. \ref{fig:analysis}. 

\begin{figure}[!h]
    \centering
    \includegraphics[width=\textwidth]{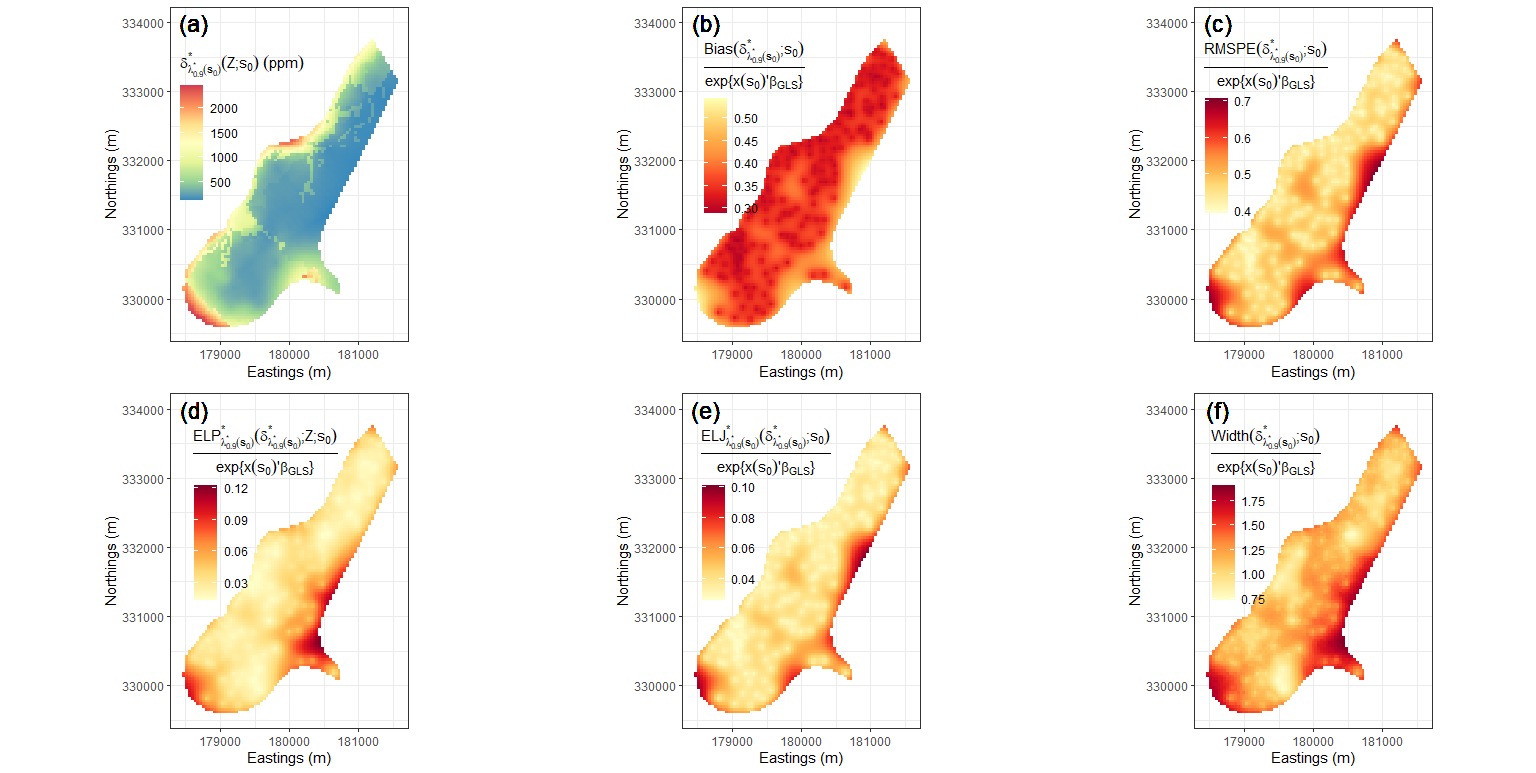}
    \caption{Maps of quantities: (a) $\delta_{\lambda^*_{0.9}(\mathbf{s}_0)}^*(\mathbf{Z}; \mathbf{s}_0)$, (b) $Bias(\delta_{\lambda^*_{0.9}(\mathbf{s}_0)}^*; \mathbf{s}_0)/\exp\{\mathbf{x}(\mathbf{s}_0)'\hat{\bm\beta}_{GLS}\}$, (c) $RMSPE(\delta_{\lambda^*_{0.9}(\mathbf{s}_0)}^*; \mathbf{s}_0)/\exp\{\mathbf{x}(\mathbf{s}_0)'\hat{\bm\beta}_{GLS}\}$, (d) $ELP_{\lambda^*_{0.9}(\mathbf{s}_0)}^*(\delta_{\lambda^*_{0.9}(\mathbf{s}_0)}^*; \mathbf{Z}; \mathbf{s}_0)/\exp\{\mathbf{x}(\mathbf{s}_0)'\hat{\bm\beta}_{GLS}\}$, (e) $ELJ_{\lambda^*_{0.9}(\mathbf{s}_0)}^*(\delta_{\lambda^*_{0.9}(\mathbf{s}_0)}^*; \mathbf{s}_0)/\exp\{\mathbf{x}(\mathbf{s}_0)'\hat{\bm\beta}_{GLS}\}$, and (f) the width of the 95\% unconditional interval divided by $\exp\{\mathbf{x}(\mathbf{s}_0)'\hat{\bm\beta}_{GLS}\}$, where $\mathbf{s}_0$ ranges over $D$.}
    \label{fig:adaptive_analysis}
\end{figure}

\section{Discussion and conclusions}\label{sec:discussion}

In this article we have presented properties of optimal spatial prediction and uncertainty quantification for positive-valued spatial processes using the family of (asymmetric) power-divergence loss functions. First, we developed a new way to quantify asymmetry of loss for positive-valued spatial processes. Then, we revisited the notion that spatial-prediction problems are also decision problems, which use loss functions; this has been known for some time \cite[]{Cressie1993}, but the implications have not been developed.  Examining spatial prediction through a decision-theoretic lens immediately questions the ubiquitous use of squared-error loss (SEL), because loss functions should quantify the consequences of over-prediction and under-prediction. In nature, these are rarely symmetric, contrary to what squared-error loss dictates. The class of power-divergence loss (PDL) functions was shown to provide an attractive alternative to SEL, not only because its members capture situations where under-prediction incurs higher losses than over-prediction and \textit{vice versa}, but also because they are continuous, convex and, when calibrated, capable of behaving similarly to several commonly used loss functions. 

In this article we have answered a series of questions: For a given parameter $\lambda$, what is the optimal power-divergence (OPD) spatial predictor and how can it be computed efficiently? What do spatial prediction intervals look like? What do the expected losses look like under PDL? Are they capable of providing uncertainty quantification? How can an OPD block predictor be obtained from point predictors? How can $\lambda$ be chosen using statistical procedures?

Several areas of investigation are left to future research, such as the use of control variates to precisely estimate the OPD spatial predictor based on MCMC samples from the predictive distribution. Other examples include fully exploring the inferential differences between conditional and unconditional prediction intervals and interrogating the connection between the OPD spatial predictors and Box-Cox transformations. And, outside the setting of spatial prediction, power-divergence-based inference may also have applications in robust estimation. 

In addition to providing statistical procedures that enable us to make a choice of $\lambda$, we would like to develop adaptive ways to determine $\lambda$. Spatial predictors that vary in type throughout the spatial domain $D$ have been considered by \cite{Bradley2015} and, similarly, OPD spatial predictors for different $\lambda$'s could be used to predict the hidden spatial field. One example of adaptive spatial prediction was considered in Section \ref{sec:calibration}, where it was shown that the power parameter $\lambda(\mathbf{s}_0)$ at prediction location $\mathbf{s}_0$ could be calibrated to allow $\delta_{\lambda(\mathbf{s}_0)}^*(\mathbf{Z}; \mathbf{s}_0)$ to predict a specified quantile $q$ of the predictive distribution. Noting that the desired quantile can be allowed to vary spatially over different parts of the domain, makes this approach truly adaptive. There may be other informative ways of choosing a spatially varying power parameter $\lambda(\mathbf{s}_0)$ as $\mathbf{s}_0$ ranges over $D$. 

As well as the intuitive justification for biased prediction given in Section \ref{sec:introduction}, we would like to develop an analysis that illustrates the entire pipeline from developing or calibrating a loss function to real-life losses, providing biased predictions to make conservative decisions, and finally assessing the benefit through the minimised expected losses. We believe that the PDL function could be used to infer spatial exceedances, and preliminary work on this is currently underway.

\section*{Acknowledgements}

This research was supported by Australian Research Council Discovery Project DP190100180 (AP and NC). We thank Andrew Zammit-Mangion for perceptive comments and suggestions and Professor Alfred Stein for his help in obtaining a copy of the thesis of \cite{Rikken1993}. The Sydney Business School, University of Wollongong generously provided an environment where the article took shape and evolved into a manuscript.

\newpage
\appendix

\numberwithin{equation}{section}
\renewcommand{\theequation}{\thesection\arabic{equation}}

{\Large \noindent\textbf{Online Supplement for ``Optimal prediction of positive-valued spatial processes: asymmetric power-divergence loss'' by Alan R. Pearse, Noel Cressie, and David Gunawan}}

\section{Proofs and derivations}\label{sec:proofs}

\subsection{Behaviour of the asymmetry measure $A_{PDL,\lambda}(f)$ as $f \rightarrow 0$ and $f \rightarrow 1$}\label{sec:asymmetry_measure_result}

For $f \in (0, 1)$ and $\lambda \in (-\infty, \infty)$, consider the function,
\begin{equation}
A_{PDL,\lambda}(f) = 
\begin{cases} 
\frac{(1-f)^{1-\lambda} - (1 - f) (1 + \lambda f) }{(1-f)^{\lambda + 1} + (\lambda + 1)f - 1}, &  \lambda \neq -1, 0\\
& \\
-\frac{(1-f)\log(1 - f) + f(1-f)}{(1-f)\log(1 - f)+f}, & \lambda = 0\\
&\\
-\frac{(1-f)^2\log(1-f) + f(1-f)}{\log(1-f)+f}, & \lambda = -1.
\end{cases}\label{eqn:A_PDL_appendix}
\end{equation}
Alternatively,
\begin{equation}
A_{PDL,\lambda}(f) = (1-f)^2 \frac{\phi^+_\lambda((1-f)^{-1})}{\phi_\lambda^+(1-f)},\label{eqn:A_PDL_alt_appendix}
\end{equation}
where $\phi^+_\lambda(x) \equiv (\lambda(\lambda + 1)^{-1})\{x^{\lambda + 1} - x + \lambda(1-x)\}$, and the function satisfies $\phi_\lambda^+(1) = (\phi_\lambda^{+})'(1) = 0$, $(\phi_\lambda^{+})''(1) > 0$, $\phi_\lambda^+(0/0) = 0$, and $0\cdot \phi_\lambda^+(p/0) = p \cdot \lim_{u\rightarrow \infty}\phi_\lambda^+(u)/u$. 

\begin{proposition}
$\lim_{f\rightarrow 0}A_{PDL,\lambda}(f) = 1$. 
\end{proposition}

\begin{proof}
   The limit $\lim_{f\rightarrow 0} A_{PDL,\lambda}(f)$ can be evaluated after twice applying L'H\^{o}pital's rule to \eqref{eqn:A_PDL_alt_appendix}:
$$
\lim_{f\rightarrow 0} A_{PDL, \lambda}(f) = \lim_{f\rightarrow 0} \frac{2\phi_\lambda^+((1-f)^{-1}) - 2(1-f)^{-1}(\phi_\lambda^{+})'((1-f)^{-1}) + (\phi^{+}_\lambda)''((1-f)^{-1})}{(\phi^{+}_\lambda)''(1-f)}.
$$
Recall the function $\phi^+_\lambda$ has the properties that $\phi^+_\lambda(1) = (\phi^{+}_\lambda)'(1) = 0$ and $(\phi^{+}_\lambda)''(1) > 0$. Therefore, $\lim_{f\rightarrow 0} A_{PDL,\lambda}(f) = 1$, as required. 
\end{proof}

\begin{proposition}
    $\lim_{f \rightarrow 1} A_{PDL, \lambda}(f)$ is $0$ if $\lambda < 1$, $\infty$ if $\lambda > 1$, and $1$ if $\lambda = 1$.
\end{proposition} 

\begin{proof}
Consider the following cases:
\begin{itemize}
    \item When $\lambda > 1$, the numerator of \eqref{eqn:A_PDL} goes to $\infty$ as $f\rightarrow 1$ and the denominator to $\lambda$ as $f \rightarrow 1$. 
    \item When $\lambda = 1$, $A_{PDL,1}(f) = 1$ for all $f$, so clearly $\lim_{f\rightarrow 1}A_{PDL, 1}(f) = 1$.
    \item When $\lambda < 1$ but $\lambda \neq -1, 0$, the numerator in \eqref{eqn:A_PDL} evaluates to $0$ and the denominator to $\lambda$ as $f \rightarrow 1$. 
    \item When $\lambda = 0$, $\lim_{f\rightarrow 1} A_{PDL,0}(f) = 0$ by inspection.  
    \item When $\lambda = -1$, $\lim_{f\rightarrow 1} A_{PDL,-1}(f) = 0$ by inspection. 
\end{itemize}
Hence $\lim_{f\rightarrow 1}A_{PDL,\lambda}(f) = 0$ for $\lambda < 1$, $1$ for $\lambda = 1$, and $\infty$ for $\lambda > 1$. 

\end{proof}

\subsection{Delta-method approximation of the OPD predictor}\label{sec:delta_method_appendix}

Let $Y > 0$ be a random variable of interest; let $Z > 0$ be data; let $-\infty < \lambda < \infty$ be a specified power parameter; and let $\delta_\lambda^*(Z)$ be an optimal predictor of $Y$ that depends on $\lambda$ and data $Z$. Note that the spatial notation has been dropped for convenience. Then, consider the optimal power-divergence (OPD) predictor of $Y$:
\begin{equation}
 \delta_\lambda^*(Z) = \begin{cases}
     E(Y^{\lambda + 1}\mid Z)^{1/(\lambda + 1)}, & \lambda \neq -1,\\
     \exp\{E(\log(Y)\mid Z)\}, & \lambda = -1.
 \end{cases}   
\end{equation}
In what follows, we derive the delta-method approximations of the OPD predictor.

For $\lambda \neq -1$, taking the second-order Taylor-series approximation of $Y^{\lambda + 1}$ about $E(Y\mid Z)$ yields,
\begin{align}
    E(Y^{\lambda + 1}\mid Z)^{1/(\lambda + 1)} &\simeq E\Bigg(E(Y\mid Z)^{\lambda + 1} + (\lambda + 1)E(Y\mid Z)^\lambda(Y - E(Y\mid Z)) \nonumber\\ &~~~~~~~~~~~~~~~~~~~~~~~~+\frac{\lambda (\lambda + 1)E(Y\mid Z)^{\lambda-1}(Y - E(Y\mid Z))^2}{2}\Bigg| Z\Bigg)^{1/(\lambda + 1)} \nonumber\\
    &\simeq \left(E(Y\mid Z)^{\lambda + 1} + \frac{\lambda(\lambda + 1)E(Y\mid Z)^{\lambda-1}\mathrm{var}(Y \mid Z)}{2}\right)^{1/(\lambda + 1)}\nonumber \\
    &\simeq E(Y\mid Z)\cdot \left(1 + \frac{\lambda(\lambda + 1)\mathrm{CV}(Y\mid Z)^2}{2}\right)^{1/(\lambda + 1)},\label{eqn:first_approx}
\end{align}
where $\mathrm{CV}(Y\mid Z) \equiv \sqrt{\mathrm{var}(Y\mid Z)}/E(Y\mid Z)$. Taking a second-order Taylor-series approximation of $x^{1/(\lambda + 1)}$ about $x=1$ yields $1 + (\lambda + 1)^{-1}(x-1) - 0.5\lambda(\lambda + 1)^{-2}(x - 1)^2$, so substituting this into \eqref{eqn:first_approx} yields,
\begin{equation}
    \delta_\lambda^*(Z) \simeq E(Y\mid Z)\left(1 + \frac{\lambda \mathrm{CV}(Y\mid Z)^2}{2} - \frac{\lambda^3\mathrm{CV}(Y\mid Z)^4}{8}\right).\label{eqn:second_approx}
\end{equation}
When $\lambda = -1$, the OPD predictor is $\delta_{-1}^*(Z) = \exp\{E(\log(Y)\mid Z)\}$. Taking the second-order Taylor-series approximation of $\log(x)$ around $x = E(Y\mid Z)$ yields,
\begin{align}
    \exp\{E(\log(Y)\mid Z)\} &\simeq \exp\Bigg\{E\Bigg(\!\log(E(Y\mid Z)) + \frac{(Y - E(Y\mid Z))}{E(Y\mid Z)}- \frac{(Y - E(Y\mid Z))^2}{2E(Y\mid Z)^2}\Bigg| Z \Bigg)\Bigg\},\nonumber\\
    &\simeq E(Y\mid Z) \cdot \exp\Bigg\{-\frac{\mathrm{CV}(Y \mid Z)^2}{2}\Bigg\}.\label{eqn:first_approx_neg1}
\end{align}
Using the well known second-order Taylor-series approximation of $\exp\{x\}$ about $x = 0$ in \eqref{eqn:first_approx_neg1} yields,
\begin{equation}
    \delta_{-1}^*(Z) \simeq E(Y\mid Z) \cdot \left(1 - \frac{\mathrm{CV}(Y\mid Z)^2}{2} + \frac{\mathrm{CV}(Y\mid Z)^4}{8}\right),\label{eqn:second_approx_neg1}
\end{equation}
which is precisely \eqref{eqn:second_approx} with $\lambda = -1$ substituted in.

\subsection{Biasedness of the OPD spatial predictor}\label{sec:bias_appendix}

We maintain the non-spatial notation given in Section \ref{sec:asymmetry_measure_result}. The OPD predictor for random variable $Y > 0$ based on data $Z > 0$ is,
$$
 \delta_\lambda^*(Z) = \begin{cases}
     E(Y^{\lambda + 1}\mid Z)^{1/(\lambda + 1)}, & \lambda \neq -1,\\
     \exp\{E(\log(Y)\mid Z)\}, & \lambda = -1.
 \end{cases}
 $$

\begin{proposition}
    The OPD predictor is biased (i.e., $E(\delta_\lambda^*(Z) - Y) \neq 0$) except for $\lambda = 0$, when it is unbiased. Furthermore, it is positively biased (i.e., $E(\delta_\lambda^*(Z)- Y) > 0$) when $\lambda > 0$ and negatively biased ($E(\delta_\lambda^*(Z) - Y) < 0$) when $\lambda < 0$.
\end{proposition}

\begin{proof}
First, with a slight abuse of notation, denote the OPD predictor defined over a range of values of $\lambda$, say $a < \lambda < b$, as $\delta^*_{(a < \lambda < b)}(Z)$. That is, the subscript denotes the range of $\lambda$ for which an OPD predictor is considered. Then, consider the following cases:
\begin{itemize}
\item When $\lambda > 0$, the function $g(x) = x^{\lambda + 1},~x>0$ is convex, so Jensen's inequality yields $E(Y^{\lambda + 1}\mid Z) \geq E(Y\mid Z)^{\lambda + 1}$; hence, $\delta_{(\lambda > 0)}^*(Z) = E(Y^{\lambda + 1}\mid Z)^{1/(\lambda + 1)} \geq E(Y\mid Z)$. 
\item When $\lambda = 0$, it is trivial to see that $\delta_0^*(Z) = E(Y\mid Z)$ is unbiased. 
\item When $-1 < \lambda < 0$, let $a = \lambda + 1$ and note that $0 < a < 1$. Then the function $g(x) = x^{a},~x > 0$ is concave. By Jensen's inequality, $E(Y^{a}\mid Z) \leq E(Y\mid Z)^{a}$; therefore, $\delta_{(-1 < \lambda < 0)}^*(Z) = E(Y^{\lambda + 1}\mid Z)^{1/(\lambda + 1)} \leq E(Y \mid Z)$.
\item When $\lambda = -1$, the function $g(x) = \log(x),~x>0$ is concave; hence Jensen's inequality yields, $E(\log(Y)\mid Z) \leq \log(E(Y\mid Z))$. Therefore, $\delta_{-1}^*(Z) = \exp\{E(\log(Y)\mid Z)\} \leq E(Y\mid Z)$. 
\item When $\lambda < -1$, let $a = |\lambda + 1|$. The function $g(x) = x^{-a}$ is convex; hence $E(Y^{-a}\mid Z) \geq E(Y \mid Z)^{-a}$ by Jensen's inequality. Equivalently, $E(Y^{-a}\mid Z)^{-1} \leq E(Y \mid Z)^{a}$ and, therefore, $\delta_{(\lambda < -1)}^*(Z) = E(Y^{-a}\mid Z)^{-1/a} \leq E(Y \mid Z)$. 
\end{itemize}
Altogether, these cases show that $E(\delta_\lambda^*(Z)- Y) < 0$ when $\lambda < 0$, $=0$ when $\lambda = 0$, and $> 0$ when $\lambda > 0$.     
\end{proof}

\section{Prediction intervals}\label{sec:prediction_intervals_appendix}

\subsection{The existence of the prediction interval under power-divergence loss}\label{sec:existence_appendix}

Let $L_{PDL,\lambda}(\delta(Z), Y)$ be the power-divergence loss with fixed $-\infty < \lambda < \infty$, for predictand $Y > 0$ and predictor $\delta(Z) > 0$. Non-spatial notation has been used for convenience. The following defines a $(1-\alpha)\times 100\%,~\alpha \in (0, 1)$ prediction interval for $Y$. Let $K_\lambda^{(\alpha)}$ be a cut-off that satisfies,
$$
\mathrm{Pr}\!\left(L_{PDL,\lambda}(\delta(Z), Y) \leq K_\lambda^{(\alpha)}\right) = 1 - \alpha.
$$
The solution, which depends on $\lambda$, is the set,
\begin{eqnarray}
 &\Big\{Y:Y^{\lambda + 1} - (\lambda + 1)\delta(Z)^\lambda Y - \lambda\delta(Z)^\lambda\!\left((\lambda + 1)K_\lambda^{(\alpha)} - \delta(Z)\right) \leq 0\Big\}&\lambda\neq 0,-1,\label{eqn:pred_int_appendix}\\
 &\Big\{Y:Y \cdot \log(Y/\delta(Z)) - (Y - \delta(Z)) - K_0^{(\alpha)} \leq 0\Big\} &\lambda = 0,\label{eqn:pred_int_lambda0_appendix}\\
 &\Big\{Y:Y - \delta(Z) - \delta(Z)\cdot \log(Y/\delta(Z)) - K_{-1}^{(\alpha)} \leq 0\Big\} &\lambda = -1.\label{eqn:pred_int_lambdaneg1_appendix}
\end{eqnarray}
Now, let $(l_\lambda, u_\lambda)$ be the bounds of the $(1-\alpha)\times 100\%$ prediction interval, with $0 \leq l_\lambda < \delta(Z)$ and $\delta(Z) < u_\lambda < \infty$. They are related to solutions of the equations,
\begin{align}
 p_{\lambda + 1}(y) \equiv y^{\lambda + 1} - (\lambda + 1)\delta(Z)^\lambda y - \lambda\delta(Z)^\lambda\!\left((\lambda + 1)K_\lambda^{(\alpha)} - \delta(Z)\right)  &= 0,&\lambda \neq -1,0,\label{eqn:pseudopolynomial}\\
 p_{1}(y) \equiv y \cdot \log(y/\delta(Z)) - (y - \delta(Z)) - K_{0}^{(\alpha)} &= 0,&\lambda= 0,\label{eqn:pseudo_0}\\
 p_{0}(y) \equiv y - \delta(Z) - \delta(Z)\cdot\log(y/\delta(Z)) - K_{-1}^{(\alpha)} &= 0, & \lambda =-1,\label{eqn:pseudo_neg1}
\end{align}
 where $y > 0$.

\begin{proposition}\label{prop:existence}
    The equation $p_{\lambda + 1}(y) = 0,$ for $-\infty < \lambda < \infty$, always has exactly one solution in the interval $(\delta(Z), \infty)$.
\end{proposition}

\begin{proof}
The functions in \eqref{eqn:pseudopolynomial}-\eqref{eqn:pseudo_neg1} are continuous over $0 < y < \infty$, and their derivatives with respect to $y$ are,
\begin{align}
p_{\lambda + 1}'(y) = (\lambda + 1)y^\lambda - (\lambda + 1)\delta^\lambda, &~\lambda \neq -1, 0\label{eqn:general_derivative}\\
p_1'(y) = \log(y) - \log(\delta(Z)), &~\lambda = 0\\
p_0'(y) = 1 - \delta(Z)/y, &~\lambda = -1.\label{eqn:neg1_derivative}
\end{align}
There is clearly a stationary point at $y = \delta(Z)$, for $-\infty < \lambda < \infty$, with values $p_{\lambda + 1}(\delta(Z)) = -\lambda(\lambda + 1)\delta(Z)^\lambda K_\lambda^{(\alpha)}$ when $\lambda \neq -1, 0$, $p_1(\delta(Z)) = -K_0^{(\alpha)}$ when $\lambda = 0$, and $p_0(\delta(Z)) = -K_{-1}^{(\alpha)}$ when $\lambda = -1$. For $-1 \leq \lambda \leq 0$, the second derivative at the stationary point indicates that the stationary point is a maximum; and for $\lambda \not\in [-1, 0]$, the stationary point is a minimum. Now, for $\delta(Z) < y < \infty$, note that the derivatives in \eqref{eqn:general_derivative}--\eqref{eqn:neg1_derivative} are always positive when $\lambda \in [-1, 0]$ (therefore $p_{\lambda + 1}(y)$ is increasing away from its minimum over the interval $(\delta(Z), \infty)$) and always negative when $\lambda \not\in [-1, 0]$ (therefore decreasing away from its maximum over the interval $(\delta(Z), \infty)$). This monotonicity for $y > \delta(Z)$, guarantees that there exists a unique $y$ in the interval $(\delta(Z), \infty)$ that satisfies $p_{\lambda + 1}(y) = 0,~-\infty < \lambda < \infty$.
\end{proof}

\subsection{Quadratic case ($\lambda = 1$ in \eqref{eqn:pseudopolynomial})}\label{sec:appendix_lambda1}

We retain the notation from Section \ref{sec:existence_appendix}. However, instead of treating \eqref{eqn:pseudopolynomial} as a function with domain $(0, \infty)$ and range $(-\infty, \infty)$, we briefly consider \eqref{eqn:pseudopolynomial} to be a function with domain $\mathbb{C}$ (the complex plane) and range $\mathbb{C}$. 

When $\lambda = 1$ in \eqref{eqn:pseudopolynomial}, $(l_1, u_1)$ are related to the roots of the quadratic,
\begin{equation}
 p_2(y) = y^{2} - 2\delta(Z) y - \delta(Z)\!\left(2K_1^{(\alpha)} - \delta(Z)\right).\label{eqn:quadratic_polynomial}   
\end{equation}
The quadratic in \eqref{eqn:quadratic_polynomial} has two real roots because its discriminant is $8\delta(Z) K_1^{(\alpha)} > 0.$ Therefore, the two real-valued roots of \eqref{eqn:quadratic_polynomial} can be seen to be, 
$$
\delta(Z) \pm \sqrt{2\delta(Z) K_1^{(\alpha)}}.
$$
One of these roots will always fall in the interval $(\delta(Z), \infty)$ and is $u_1$. However, the smaller root may be negative, in which case $l_1 = 0$. Therefore, the $(1-\alpha)\times 100\%$ prediction interval is given by,
$$
\left(\max\!\left\{0, \delta(Z) - \sqrt{2\delta(Z) K_1^{(\alpha)}}\right\}, \delta(Z) + \sqrt{2\delta(Z) K_1^{(\alpha)}}\right).
$$

\subsection{Cubic case ($\lambda = 2$ in \eqref{eqn:pseudopolynomial})}\label{sec:appendix_lambda2}

When $\lambda = 2$ in \eqref{eqn:pseudopolynomial}, the bounds $(l_2, u_2)$ are related to the real-valued roots of the depressed cubic polynomial, 
\begin{equation}
p_3(y) = y^{3} - 3\delta(Z)^2 y - 2\delta(Z)^2\!\left(3K_2^{(\alpha)} - \delta(Z)\right).\label{eqn:cubic_polynomial}
\end{equation}
All the coefficients in \eqref{eqn:cubic_polynomial} are real-valued. The discriminant for \eqref{eqn:cubic_polynomial} is $324\delta(Z)^4K_2^{(\alpha)}(2\delta(Z) - 3K_2^{(\alpha)}).$ When $\delta(Z) > 1.5K_2^{(\alpha)}$, the discriminant is positive and there are three distinct real-valued roots. When $\delta(Z) < 1.5K_2^{(\alpha)}$, the discriminant is negative and there is only one real-valued root. When $\delta(Z) = 1.5K_2^{(\alpha)}$, the discriminant is zero; together with the nature of the coefficients (i.e., they are real-valued and non-zero when $\delta(Z) = 1.5K_2^{(\alpha)}$), this implies there are two distinct real-valued roots, with one of these being a double root. As established through the analysis of \eqref{eqn:pseudopolynomial}, there is always one real-valued root in the interval $(\delta(Z), \infty)$, for all $\lambda$. Therefore, a valid one-sided prediction interval can always be defined, regardless of the nature of the roots implied by the discriminant. Our analysis concentrates on finding $l_2$ using the general formula for the roots of a depressed cubic equation in \eqref{eqn:cubic_polynomial}, from which we obtain, for $k = 0, 1, 2$,
\begin{align}
   r_k &= \psi^k\sqrt[3]{\delta(Z)^2\left(3K_2^{(\alpha)} - \delta(Z) + \sqrt{9\left(K^{(\alpha)}_2\right)^2 - 6\delta(Z) K_{2}^{(\alpha)}}\right)}\nonumber\\ 
   &~~~~~~~~~~~~+ \delta(Z)^2\Bigg/\psi^k\sqrt[3]{\delta(Z)^2\left(3K_2^{(\alpha)} - \delta(Z) + \sqrt{9\left(K^{(\alpha)}_2\right)^2 - 6\delta(Z) K_{2}^{(\alpha)}}\right)},\label{eqn:cubic_solution}
\end{align}
where $\psi = (-1 + \sqrt{-3})/2$ is a primitive cube root of unity. If the discriminant is negative and \eqref{eqn:cubic_polynomial} has two complex-valued roots, the prediction interval is $(0, u_2)$, where $u_2$ is the only real-valued root of \eqref{eqn:cubic_polynomial}. When the discriminant is zero (i.e., $\delta(Z) = 1.5K_2^{(\alpha)}$), \eqref{eqn:cubic_solution} yields a double root at $-1.5K_2^{(\alpha)}$ and another root at $3K_2^{(\alpha)}$; therefore, the prediction interval is $(0, 3K_2^{(\alpha)})$. Otherwise, if the discriminant is positive, suppose the roots are ordered such that $r_{(0)} < r_{(1)} < r_{(2)}$ Then, the prediction interval is given by $(l_2, u_2) = (\max\{0, r_{(1)}\}, r_{(2)})$.  

\subsection{Quartic case ($\lambda = 3$) and higher degrees}\label{sec:appendix_lambda3}

When $\lambda = 3$, \eqref{eqn:pseudopolynomial} is a quartic polynomial. With some effort, the explicit formulae for the solutions of \eqref{eqn:pseudopolynomial} with $\lambda = 3$ can be derived, but they are are not provided here. Cases involving $\lambda = 4, 5, ...$ are not pursued due to the following proposition. 

\begin{proposition}
    Explicit formulae for the solutions of \eqref{eqn:pseudopolynomial} do not exist for integer values of $\lambda > 3$.
\end{proposition} 

\begin{proof}
    This immediately follows from the Abel-Ruffini theorem \cite[]{Ruffini1813, Abel1826}, which states that there are no solutions in terms of radicals (i.e., $n$-th roots for $n > 1$) for polynomials of degree $5$ or higher. When $\lambda$ is a positive integer, the degree of the polynomial in \eqref{eqn:pseudopolynomial} is $\lambda + 1$, so this limitation holds for any integer value of $\lambda > 3$. 
\end{proof}

\subsection{A simple algorithm to find the prediction interval for any $\lambda$}\label{sec:prediction_interval_algorithm_appendix}

Let $\delta_\lambda^*(Z)$, with fixed $\lambda \in (-\infty, \infty)$, be the optimal power-divergence predictor of random variable $Y > 0$ based on data $Z > 0$. The following is a simple algorithm to solve for $(l_\lambda, u_\lambda)$ in \eqref{eqn:pseudopolynomial}, \eqref{eqn:pseudo_0}, or $\eqref{eqn:pseudo_neg1}$: 
\begin{enumerate}
    \item As determined by $\lambda$, search the interval $(0, \delta_\lambda^*(Z))$ for a solution to \eqref{eqn:pseudopolynomial}, \eqref{eqn:pseudo_0}, or \eqref{eqn:pseudo_neg1} using (for example) a bisection algorithm. There may not be a solution in the interval $(0, \delta_\lambda^*(Z))$, in which case $l_\lambda = 0$, and the interval becomes one-sided. 
    \item As determined by $\lambda$, search the interval $(\delta_\lambda^*(Z), \infty)$ for a solution to \eqref{eqn:pseudopolynomial}, \eqref{eqn:pseudo_0}, or \eqref{eqn:pseudo_neg1}. If the bisection algorithm is used, initialise the search on the two endpoints $\delta_\lambda^*(Z)$ and a large multiple of $\delta_\lambda^*(Z)$; if the function on the left-hand side of \eqref{eqn:pseudopolynomial}, \eqref{eqn:pseudo_0}, or \eqref{eqn:pseudo_neg1} does not have opposite signs at the proposed endpoints, increase the multiple. Repeat until a solution is found. Recall that Proposition \ref{prop:existence} guarantees the existence of a solution.
\end{enumerate}

\end{document}